\documentclass[conf]{new-aiaa}

\usepackage[utf8]{inputenc}
\usepackage{textcomp}
\usepackage{graphicx}
\usepackage{amsmath}
\usepackage{mathtools}
\usepackage[english]{babel}
\usepackage[version=4]{mhchem}
\usepackage{siunitx}
\usepackage{longtable,tabularx}
 \usepackage{amsthm}
\usepackage[utf8]{inputenc}
\usepackage{graphicx}
\usepackage{subcaption}
\usepackage{amsmath}
\usepackage[version=4]{mhchem}
\usepackage{siunitx}
\usepackage{longtable,tabularx}
\usepackage{algorithm}
\usepackage{algpseudocode}
\usepackage{amsmath}
\usepackage{float}
\usepackage{titlesec}

\usepackage{xcolor}
\usepackage{textcomp}
\usepackage{graphicx}
\usepackage{amsmath}
\usepackage[version=4]{mhchem}
\usepackage{siunitx}
\usepackage{longtable,tabularx}
\usepackage{graphicx}
\usepackage{subcaption}
\usepackage{amsmath}
\usepackage{mathtools}
\usepackage[version=4]{mhchem}
\usepackage{siunitx}
\usepackage{longtable,tabularx}
\usepackage{algorithm}
\usepackage{algpseudocode}
\usepackage{amsmath}
\usepackage{amsfonts}
\usepackage{comment}
\usepackage{titlesec}
\usepackage{booktabs,subcaption,amsfonts,dcolumn}
\usepackage{xcolor}

\newtheorem{problem}{Problem}
\newtheorem{prop}{Proposition}

\title{
 Minimum-fuel Spacecraft Rendezvous based on Sparsity Promoting Optimization}

\author{Vrushabh Zinage\footnote{Graduate Student, Department of Aerospace Engineering and Engineering Mechanics, Student Member AIAA} and Efstathios Bakolas \footnote{Associate Professor, Department of Aerospace Engineering and Engineering Mechanics email: bakolas@austin.utexas.edu Senior Member AIAA}}
\affil{Department of Aerospace Engineering and Engineering Mechanics, University of Texas at Austin, Austin, Texas 78712}

\begin{document}

\maketitle

 \begin{abstract}
In this paper, we consider the classical spacecraft rendezvous problem in which the so-called active spacecraft has to approach the target spacecraft which is moving in an elliptical orbit around a planet by using the minimum possible amount of fuel. Instead of using standard convex optimization tools which can be computationally expensive, we use modified versions of the Iteratively Reweighted Least Squares (IRLS) algorithm from compressive sensing to compute sparse optimal control sequences which minimize the fuel consumption for both thrust vectoring and orthogonal vectoring (active) spacecraft. Numerical simulations are performed to verify the efficacy of our approach. 
 \end{abstract}

 \section*{Nomenclature}


 {\renewcommand\arraystretch{1.0}
 \noindent\begin{longtable*}{@{}l @{\quad=\quad} l@{}}
 $(x,y,z)$  & relative position of the active chaser spacecraft w.r.t the target spacecraft \\
 $\boldsymbol{u}$  & control input $[u_x\;u_y\;u_z]^\mathrm{T}$ of the active chaser spacecraft in the LVLH frame\\
 $e$ &    eccentricity of the target elliptical orbit \\
 $\nu$& true anomaly \\
 $\omega$ & orbit rate of the target \\
  $h$ & angular momentum of the target orbit \\
  $\mu$ & gravity constant \\
   $r$ & distance of active chaser spacecraft from the target spacecraft \\
   $\gamma$ & $\mu/h^{3/2}$ \\
 $a$ & semi-major axis of target orbit  \\
 $s$ & $\rho\text{sin}(\nu)$ \\
 $c$ & $\rho\text{cos}(\nu)$ \\
 $(\cdot)^{\mathrm{T}}$ &  transpose of matrix\\
 $\rho$ &  1+$e$cos($\nu$) \\
 $\nu_o$ &  initial true anomaly of active chaser spacecraft\\
 $\nu_f$ &  final true anomaly of active chaser spacecraft\\
 $\Phi$ &  state transition matrix\\
 $[0,N-1]_d$ & $\{0,1, \dots, N-1\}$ {(discrete interval from $0$ to $N-1$)}\\
 $\mathbf{I}_n$ & $n\times n$ identity matrix
 \end{longtable*}}

\section{Introduction}

This paper deals with the computation of (approximations of)  minimum-fuel control manoeuvres for spacecraft proximity operations based on sparsity-promoting optimization.
A major challenge in space proximity operations is to achieve autonomy for spacecraft with limited computational resources performing far-range rendezvous along an elliptical orbit while ensuring minimum fuel consumption. The far-range rendezvous is an orbital transfer between an active chaser spacecraft and a passive target spacecraft, with specified initial and final positions and velocities, over a fixed time period. This class of spacecraft maneuvers have played a key role in various space missions such as Vostok, Gemini and Apollo, \cite{chamberlin1964gemini_1,burton1966gemini_2,goodman2006history} for on-orbit satellite servicing and formation flight. In all these applications, minimizing fuel consumption is of prime importance because the amount of propellant carried by the spacecraft is severely limited. The main motivation of this work is to provide solutions which do not require the use of sophisticated and computationally expensive methods but instead, they are easily implementable by non-experts and have minimal hardware and software requirements. {To this aim, we propose the utilization of algorithms from compressive (or compressed) sensing (CS) which are  computationally efficient and easily implementable. CS is an active field of research at present that is attracting considerable interest (primarily in the signal processing community) and is widely used in signal transmission, compression, and recovery. Tools from compressive sensing can also be used to generate sparse optimal control input signals as shown in \cite{sparse_bakolas2019computation}.}

{\textit{Literature review:}} Spacecraft maneuvering can be  mainly operated in two control modes, namely, \textit{orthogonal vectoring} and \textit{thrust vectoring} \cite{leomanni2019sum}. The ability to generate thrust in any direction, that is, along yaw, roll and pitch is called \textit{thrust vectoring}. In thrust vectoring, there is a single movable thruster which is controlled using reaction wheels via attitude control. Using thrust vectoring leads to spacecraft with reduced mass and also   allows for volume savings in the thrusters. For these reasons, thrust vectoring is particularly suitable for the next generation of nano and micro spacecrafts. In the case of \textit{orthogonal vectoring}, there are three fixed thrusters along the three coordinate axes of the Local-Vertical-Local-Horizontal (LVLH) frame. It turns out that in the case of orthogonal vectoring, fuel consumption is directly proportional to the $\ell_1$ norm of the control sequence and the $\ell_2/\ell_1$ norm in the case of thrust vectoring~\cite{leomanni2019sum,zinage2021far}. The authors of \cite{massioni2011matching_pursuit} use matching pursuit and orthogonal matching pursuit algorithms to generate (approximations) of the sparsest control sequences (that is, control sequences comprised of the smallest possible number of non-zero elements) which will keep the output tracking error within certain bounds for a given reference signal. The approach in \cite{massioni2011matching_pursuit} requires that a reference trajectory is known and in addition, the terminal time is assumed to be free.
Numerical methods that are based on the primer vector theory are presented in \cite{lion1968primer_calculus_of_variations,arzelier2013new_conv4,prussing1969illustration}. The algorithms proposed in \cite{lion1968primer_calculus_of_variations} use a penalty minimization step to find the optimal number of impulses required to generate a smooth optimal trajectory. The algorithms proposed in \cite{arzelier2013new_conv4} rely on variational methods combined with polynomial optimization tools, whose complexity and computational cost, however, make them hard to apply in practical problems. Other sophisticated numerical techniques for minimum-fuel trajectory optimization based on the solution of minimum-$\ell_1$ norm problems are proposed in \cite{prussing1969optimal_conv2,carter1995linearized_conv3,arzelier2013new_conv4}. However, \cite{lion1968primer_calculus_of_variations,prussing1969optimal_conv2,carter1995linearized_conv3,arzelier2013new_conv4} may not always offer convergence guarantees, their computationally cost can be high and are not easily implementable by the non-expert. References \cite{ping_lu2013autonomous_cone_2,liu2013robust,ping_liu2014solving,p:acikmese2017} consider more general and challenging proximity operation problems under realistic constraints. The solution approaches proposed in these references offer convergence guarantees but rely on sophisticated optimization tools (e.g., interior-point methods) which have considerable cost and are not easily implementable by the non-expert.

Various algorithms have been proposed for space proximity operation problems in the literature \cite{huber_karlgaard2006robust_huber,adpative_singla2006adaptive,d_amico_2013noncooperative,neural_youmans1998neural,gao2009multi_h_infinity,yao2010flyaround_relative,luo2007optimal}. In particular, \cite{huber_karlgaard2006robust_huber} proposes a Huber filter approach to the spacecraft rendezvous problem using radar based navigation, whereas adaptive control methods for docking and rendezvous problems are proposed in \cite{adpative_singla2006adaptive}. In \cite{d_amico_2013noncooperative}, spacecraft proximity operations are performed using Global Positioning System and Optical Navigation (ARGON). Neural network based controllers \cite{neural_youmans1998neural} and multi-objective robust $H_{\infty}$ control \cite{gao2009multi_h_infinity} have been proposed in the literature for spacecraft rendezvous on a circular orbit. Algorithms which are based on a relative orbit elements for spacecraft rendezvous problems in which the target vehicle is assumed to be either cooperative or non-cooperative are proposed in~\cite{yao2010flyaround_relative}. A multi-objective optimization approach to the linearized impulsive rendezvous problem is proposed in \cite{luo2007optimal}. However \cite{huber_karlgaard2006robust_huber,adpative_singla2006adaptive,neural_youmans1998neural,gao2009multi_h_infinity,yao2010flyaround_relative,luo2007optimal} are computationally expensive and do not guarantee any optimality in terms of fuel consumption.

{\textit{Contributions:}} It is well known \cite{prussing1969optimal_conv2} that, when the dynamics of the rendezvous problem can be approximated by an autonomous or non-autonomous system of linear differential equations, minimum-fuel problems can be formulated as convex programs. Thus, one can address this class of problems by utilizing standard convex optimization techniques \cite{boyd2004convex,grant2014cvx_boyd}. However, in many space proximity operations, the spacecraft may have very limited computational resources and therefore, the control algorithms executed on-board such vehicles should be robust and have a small computational cost. In this paper, computationally light-weight algorithms that solve the minimum-fuel rendezvous problem along elliptical orbits by computing control sequences with minimum $\ell_2/\ell_1$ norm (for the thrust vectoring case) and minimum $\ell_1$ norm solutions (for the orthogonal vectoring case) are proposed. The latter solutions are computed by means of two modified versions of the Iteratively Reweighted Least Squares (IRLS) algorithm, an iterative algorithm from  compressive (or compressed) sensing~\cite{b:Foucart2013} which generates a sequence of minimizers of corresponding quadratic programs which are computationally tractable. The latter sequence converges to the minimizer of the original problem. The use of the IRLS algorithm is motivated by the fact that it can be implemented easily without requiring the use of specialized software. Furthermore, the family of IRLS algorithms are known to be robust and efficient \cite{robust_efficient_irls_2010iteratively,candes2005decoding_irls} and in addition, the solutions they generate have desired sparsity properties (sparsity promoting optimization). In particular, the IRLS algorithm allows one to compute impulsive-like, yet continuous and bounded, approximations to the required thrust inputs. To the best knowledge of the authors, this is the first paper which uses tools from compressive sensing to solve the minimum-fuel rendezvous problem for both thrust vectoring and orthogonal vectoring spacecraft.

{\textit{Structure of the Note:}} The rest of this paper is organized as follows. In Section \ref{sec:state_space_model}, the state space model of the spacecraft rendezvous is introduced. The formulation of the minimum-fuel spacecraft rendezvous problem is given in Section \ref{sec:problem}. Section \ref{sec:algorithms} introduces the modified IRLS algorithms to generate approximate minimum $\ell_2/\ell_1$ and $\ell_1$ control sequences that solve the rendezvous problem. The applicability of the proposed design is investigated through numerical simulations in Section \ref{sec:simulation_results}, and some {concluding remarks are discussed in} Section \ref{sec:conclusions}.

\section{State Space Model}\label{sec:state_space_model}
\begin{figure}[h]
\centering
\includegraphics[width=13cm]{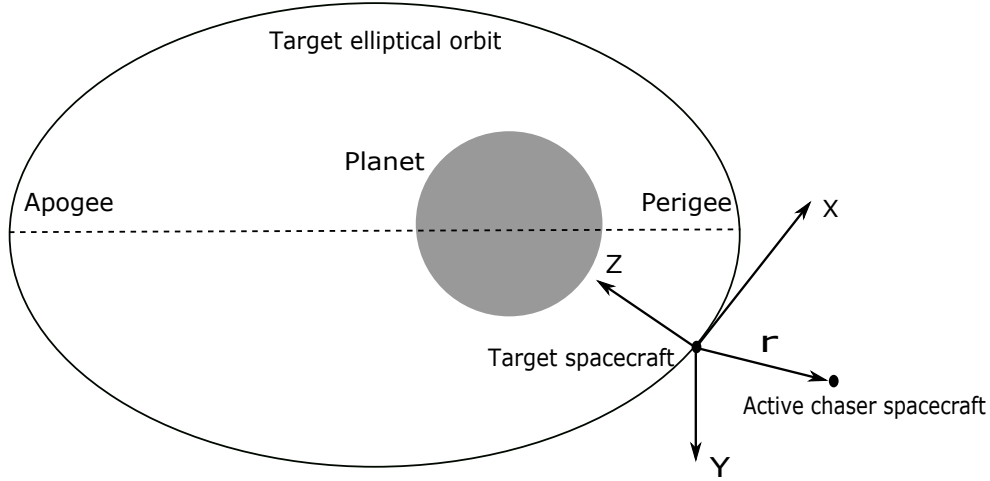}
\caption{Local-Vertical-Local-Horizontal (LVLH) coordinate system for spacecraft rendezvous}
\label{fig:coorinate_frame}
\end{figure}
Consider a target passive spacecraft moving in an elliptical orbit whose eccentricity is $e$. Let the moving Local-Vertical-Local-Horizontal (LVLH) frame be located at the center of gravity of this passive spacecraft. The relative dynamics of the active chaser spacecraft in the LVLH frame is given by \cite{yamanaka2002new_dynamics_state_trans},
  \begin{align}
  & \rho x^{\prime \prime} =2 \rho z^{\prime}-2z e \sin \nu  +2 e x^{\prime} \sin \nu +e x\cos \nu, \label{eqn:x_dynamics}\\ 
  & \rho y^{\prime \prime} =-y+2 e y^{\prime}\sin \nu, \label{eqn:y_dynamics}\\
  & \rho z^{\prime \prime} =-2 \rho x^{\prime}+2 ex \sin \nu  +2 ez^{\prime} \sin \nu +(3+e \cos \nu) z .\label{eqn:z_dynamics} 
\end{align}
where $(\cdot)'$ denotes the derivative with respect to the true anomaly $\nu$, $\rho=1+e\text{cos}(\nu)$, and $[x\;y\;z]^{\mathrm{T}}$ and $[x'\;y'\;z']^{\mathrm{T}}$ denote, respectively, the relative position and its derivative with respect to true anomaly of the active chaser in the $X,$ $Y$ and $Z$ axes 
of the Local-Vertical-Local-Horizontal (LVLH) frame as shown in Fig. \ref{fig:coorinate_frame}. Let us consider the following non-autonomous state transformation:
\begin{equation}
    [\tilde{x}\;\tilde{y}\;\tilde{z}]^{\mathrm{T}} = \rho \; [{x}\;{y}\;{z}]^{\mathrm{T}}.
    \label{eqn:change_of_coord}
\end{equation}
Now taking derivative with respect to the {true anomaly $\nu$,} Eqn.~\eqref{eqn:change_of_coord} becomes 
\begin{align}
    [\tilde{x}'\;\tilde{y}'\;\tilde{z}']^{\mathrm{T}} = \rho \; [{x}'\;{y}' \;{z}']^{\mathrm{T}}+\rho'[x\;y\;z]^\mathrm{T}.
    \label{eqn:change_of_coord_derivative}
\end{align}
Since $ [{x}\;{y}\;{z}]^{\mathrm{T}}=[\tilde{x}\;\tilde{y}\;\tilde{z}]^{\mathrm{T}}/\rho$, Eqn.~\eqref{eqn:change_of_coord} becomes
\begin{align}
    (\rho-\rho')[\tilde{x}'\;\tilde{y}'\;\tilde{z}']^{\mathrm{T}} = {\rho^2} \; [{x}'\;{y}' \;{z}']^{\mathrm{T}}.
    \label{eqn:transformation}
\end{align}
If the equation $\rho-\rho'= 0$ does not have a solution, then the transformation in Eqn. \eqref{eqn:transformation} is well defined (this is the case when $0\leq e<1/\sqrt{2}$). 
 Then, Eqns. \eqref{eqn:x_dynamics}-\eqref{eqn:z_dynamics} reduce to 
\begin{align}
    & \tilde{x}^{\prime \prime}=2 \tilde{z}^{\prime},\\
    & \tilde{y}^{\prime \prime}=-\tilde{y},\\
& \tilde{z}^{\prime \prime}=3 \tilde{z} / \rho-2 \tilde{x}^{\prime}.
\end{align}
 In the expression for the relative motion given in Eqns. (\ref{eqn:x_dynamics})-(\ref{eqn:z_dynamics}),  the true anomaly $\nu$ is used as the independent variable instead of time $t$. 
 Note that the {dynamics in the $y$-axis is decoupled from the dynamics in the $x-z$ plane}. The relative motion of the active chaser spacecraft with respect to the passive spacecraft can then be described by the following {non-autonomous} state space model \cite{yamanaka2002new_dynamics_state_trans}:
\begin{align}
    \boldsymbol{x}'(\nu)=A_c(\nu)\boldsymbol{x}(\nu)+B_c(\nu)\boldsymbol{u}(\nu),
    \label{eqn:continous_state_space_model}
\end{align}
where $\boldsymbol{x}:=[\tilde{x}\;\tilde{y}\;\;\tilde{z}\;\tilde{x}'\;\tilde{y}'\;\tilde{z}']^{\mathrm{T}}$ and

\begin{align}\label{eq:csspace}
    & A_c(\nu)=\left[\begin{array}{cccccc}0 & 0 & 0 & 1 & 0 & 0 \\ 0 & 0 & 0 & 0 & 1 & 0 \\ 0 & 0 & 0 & 0 & 0 & 1 \\ 0 & 0 & 0 & 0 & 0 & 2 \\ 0 & -1 & 0 & 0 & 0 & 0 \\ 0 & 0 & 3 /(1+e\text{cos}(\nu)) & -2 & 0 & 0\end{array}\right], & B_c(\nu)=\frac{1}{\gamma^3\rho^4}\left[\begin{array}{c}\mathbb{O}_{3 \times 3} \\ \mathbb{I}_{3}\end{array}\right].
\end{align}
Note that $A_c(\nu)$ and $B_c(\nu)$ are periodic matrix-valued functions with period $2\pi$. In practical applications, the initial conditions are given in terms of $[{x}\;{y}\;{z}\;\dot{x}\;\dot{y}\;\dot{z}]^{\mathrm{T}}$, which can be associated with the state vector $\boldsymbol{x}:=[\tilde{x}\;\tilde{y}\;\;\tilde{z}\;\tilde{x}'\;\tilde{y}'\;\tilde{z}']^{\mathrm{T}}$ via the following transformation:
\begin{align}
   [\tilde{x}\;\tilde{y}\;\;\tilde{z}\;\tilde{x}'\;\tilde{y}'\;\tilde{z}']^{\mathrm{T}}=L [{x}\;{y}\;\;{z}\;\dot{x}\;\dot{y}\;\dot{z}]^{\mathrm{T}},
\end{align}
where $L$ is given by
\begin{align}
    L=\left[\begin{array}{llllll}\rho & 0 & 0 & 0 & 0 & 0 \\0 & \rho & 0 & 0 & 0 & 0 \\ 0 & 0 &\rho & 0 &0 & 0 \\ \rho^{\prime} &0 & 0 & \frac{\rho}{\omega} & 0 &0 \\0 & \rho' & 0 & 0 & \frac{\rho}{\omega} & 0 \\ 0 &0 & \rho^{\prime} & 0 &0 &\frac{\rho}{\omega}\end{array}\right].
    \label{tansformation matrix}
\end{align}
The system in \eqref{eqn:continous_state_space_model} can be described (approximately) by the following non-autonomous discrete state space model.
\begin{align}
    \boldsymbol{x}(k+1)=A(k)\boldsymbol{x}(k)+B(k)\boldsymbol{u}(k), \quad\quad  {k\in[0,N-1]_d},
    \label{eqn:linear_system_discrete}
\end{align}
where $k$ is the stage and $\nu_k=\nu_o+\frac{(\nu_f-\nu_o)}{N}k$ is the true anomaly at stage $k$, for $k\in[0,N-1]_d$, and {the matrices $A(k)$ and $B(k)$ are defined as, respectively,
\begin{align}
    & {A}(k)=\Phi(\nu_{k+1},\nu_{k}),\\
    & {B}(k)=\int_{\nu_{k}}^{\nu_{k+1}}\Phi(\nu_{k+1},\sigma){B}_c d\sigma,
\end{align}
where $\Phi$ is the state transition matrix of the continuous state space model~\eqref{eq:csspace},  and $\alpha=\frac{(\nu_f-\nu_o)}{N}$ is the sampling period. Since the matrix $A_c$ that appears in the non-autonomous continuous state space model given in Eqn.~\eqref{eqn:continous_state_space_model} depends on $\nu$, the corresponding state transition matrix $\Phi$ does not admit an analytic expression and will have to be approximated numerically.
To this aim, we have}
\begin{align}
    \Phi(\nu,\nu_0)=\Phi_{\nu} \Phi_{\nu_{0}}^{-1},
    \label{eqn:state_transition_matrix}
\end{align}
where the matrices $\Phi_{\nu} $ and $\Phi_{\nu_{0}}^{-1}$ are given by
\begin{align}
\Phi_{\nu}=\left[\begin{array}{cccccc}1 & 0 & -c(1+1 / \rho) & s(1+1 / \rho) & 0 & 3 \rho^{2} J(\nu) \\0 & c/\rho & 0 & 0 & s/\rho & 0\\ 0 & 0 & s & c & 0 & (2-3 e s J(\nu)) \\ 0 & 0 & 2 s & 2 c-e & 0 & 3(1-2 e s J(\nu)) \\0 & -s/\rho & 0 & 0 & c/\rho & 0 \\ 0 & 0 & s^{\prime} & c^{\prime} & 0 & -3 e\left(s^{\prime} J(\nu)+s / \rho^{2}\right)\end{array}\right],
\label{eqn:phi_nu}
\end{align}
\begin{align}
  \Phi_{\nu_{0}}^{-1}=\frac{1}{1-e^{2}} &\left[\begin{array}{cccccc}1-e^{2} & 0 & 3 e(s / \rho)(1+1 / \rho) & -e s(1+1 / \rho) & 0 & -e c+2 \\ 0 & c(1-e^2)/\rho & 0 & 0 & s(1-e^2)/\rho & 0\\0 & 0 & -3(s / \rho)\left(1+e^{2} / \rho\right) & s(1+1 / \rho) & 0 & c-2 e \\ 0 & 0 & -3(c / \rho+e) & c(1+1 / \rho)+e & 0 & -s \\0 & -s(1-e^2)/\rho & 0 & 0 & c(1-e^2)/\rho & 0 \\ 0 & 0 & 3 \rho+e^{2}-1 & -\rho^{2} & 0 & e s\end{array}\right]_{\nu=\nu_{0}},
\end{align}
and $\rho=1+e\text{cos}(\nu)$, $s=\rho\text{sin}(\nu)$, $c=\rho\text{cos}(\nu)$ and $J$ in Eqn. \eqref{eqn:phi_nu} is given by 
\begin{align}
  J(\nu)=\int_{\nu_{0}}^{\nu} \frac{1}{\rho(\tau)^{2}} \mathrm{d} \tau.
\end{align}
Since $J$ does not have an analytical expression, it can be characterized numerically. From Eqn.~\eqref{eqn:state_transition_matrix}
\begin{align}
    \Phi(\nu_k,\nu_{k-1})=\Phi_{{\nu_{{k}}}} \Phi_{\nu_{k-1}}^{-1},
    \label{eqn:state_transition_matrix_discrete}
\end{align}
where $k\in[1,N]$ and $\nu_k\in[\nu_o,\nu_f]$.
Hence from Eqn. \eqref{eqn:linear_system_discrete}
it follows that the terminal state at $k=N$ is given by
\begin{align}
    \boldsymbol{x}(N)=\prod_{k=0}^{N-1} {A}(k)\boldsymbol{x}(0)+\sum_{\tau=0}^{N-1}\Big(\prod_{k=1+\tau}^{N-1} {A}(k)\Big)B(\tau)\boldsymbol{u}(\tau).
    \label{eqn:discrete_linear_sytem_final}
\end{align}
The state transition matrix of the discrete-time system \eqref{eqn:linear_system_discrete}, $\Phi_d(k,m)$, is introduced as follows:
\begin{align}
\Phi_d(k,m) =\left\{ \begin{array}{l}
A(k-1)\dots A(m)  \quad\quad k>m\geq 0
\\
\mathbf{I}_6 ~~\qquad\qquad \qquad \qquad k=m
\end{array}
\right.
\end{align}
where $k$ and $m$ are non negative integers. Therefore Eqn.~\eqref{eqn:discrete_linear_sytem_final} becomes
\begin{align}
        & \boldsymbol{x}(N)=\Phi_d(N,0)\boldsymbol{x}(0)+\sum_{\tau=0}^{N-1}\Phi_d(N,\tau+1)B(\tau) \boldsymbol{u}(\tau).
\end{align}
From Eqn. \eqref{eqn:discrete_linear_sytem_final}, the terminal state of the active spacecraft can be written in a compact form as follows
\begin{align}
    \boldsymbol{x}(N)=\boldsymbol{\beta}+\boldsymbol{C}_N\boldsymbol{U},  \quad 
    \label{eqn:terminal_constraint}
\end{align}

where 
\begin{align}
    & \boldsymbol{U}  =[\boldsymbol{u}(0)^\mathrm{T},\;\boldsymbol{u}(1)^\mathrm{T},\;\dots\boldsymbol{u}(N-1)^\mathrm{T}]^\mathrm{T},\label{eqn:control_input_U}\\
    & \boldsymbol{C}_N =[\Phi_d(N,1)B(0),\;\;\Phi_d(N,2)B(1),\dots B(N-1)],\label{eqn:cn}\\
    & \boldsymbol{\beta} =\Phi_d(N,0)\boldsymbol{x}(0).
    \label{eqn:beta}
\end{align}


\section{Problem Formulation}\label{sec:problem}
{In this section, the  minimum-fuel problem for the system described by Eqn. \eqref{eqn:linear_system_discrete} is formulated. It is assumed that the passive target and active chaser spacecraft are initially located on two non-coplanar, non-circular orbits. The chaser spacecraft is then required to satisfy terminal constraints of position and velocity of the target spacecraft at a fixed final instant of true anomaly $\nu_f$, while minimizing fuel consumption for thrust vectoring and orthogonal vectoring.}

{
In particular, let $\boldsymbol{U}_{0: N-1}= \{\boldsymbol{u}(k) \in \mathbb{R}^3: k\in[0,N-1]_d \}$ denote the sequence 
of inputs applied to the system \eqref{eqn:linear_system_discrete} for $k\in [0,N-1]_d$. As it is already mentioned, there are basically two control modes through which a spacecraft can operate; thrust vectoring and orthogonal vectoring. Thrust vector maneuvering can be achieved with a single thruster which can point in any direction in the $X-Y-Z$ coordinate axes of the LVLH frame by steering the thruster using attitude control commands. In the case of orthogonal vectoring, there are three fixed thrusters along the $X-Y-Z $ coordinate axes in the LVLH frame. Following \cite{leomanni2019sum}, it is assumed that in the thrust vectoring case, the fuel consumed is directly proportional to $\left\|\boldsymbol{U}_{0: N-1}\right\|_{\ell_{2}/\ell_1}$, where 
\begin{equation}
    \left\|\boldsymbol{U}_{0: N-1}\right\|_{\ell_{2}/\ell_1}=\sum_{i=0}^{N-1}\|\boldsymbol{u}(i)\|_2,
    \label{eqn:performance_index_l2/l1}
\end{equation}  
whereas in the orthogonal vectoring case, the fuel consumed is directly proportional to $\|\boldsymbol{U}_{0: N-1}\|_{\ell_{1}}$ where,
\begin{equation}
   \left\|\boldsymbol{U}_{0: N-1}\right\|_{\ell_{1}} = \sum_{i=0}^{N-1}\|\boldsymbol{u}(i)\|_1=\|\boldsymbol{U}\|_1.
       \label{eqn:performance_index_l1}
\end{equation}
The corresponding optimal control (minimum-fuel) problems are formulated as follows:}

\begin{problem}
{Let $\boldsymbol{x}_{0}$, $\boldsymbol{x}_f\in\mathbb{R}^6$ and $N>0$ be given. Find a control sequence $\boldsymbol{U}_{0: N-1}^{\star}:=\left\{\boldsymbol{u}^{\star}(k) \in \mathbb{R}^{3}: k\in[0, N-1]_d\right\}$ that will steer the system described by Eqn. \eqref{eqn:discrete_linear_sytem_final} from state $\boldsymbol{x}(0)=\boldsymbol{x}_{0}$ at stage $k=0$ to the final state $\boldsymbol{x}(N)=\boldsymbol{x}_f$ at stage $k=N$ while minimizing the performance index $\mathcal{J}_{2,1}\left(\boldsymbol{U}_{0: N-1}\right):=\left\|\boldsymbol{U}_{0: N-1}\right\|_{\ell_2/\ell_1}$ (thrust vectoring) or $\mathcal{J}_{1}\left(\boldsymbol{U}_{0: N-1}\right):=\left\|\boldsymbol{U}_{0: N-1}\right\|_{\ell_1}$ (orthogonal vectoring).}
\label{problem:optimal_control_problem}
\end{problem}

{Next, Problem \ref{problem:optimal_control_problem} is associated with two convex optimization problems, an $\ell_2/\ell_1$ norm minimization problem for the thrust vectoring case and an $\ell_1$-norm minimization problem for the orthogonal vectoring case.}

\begin{problem}
Find a vector $\boldsymbol{U}^\star\in \mathbb{R}^{6N}$ that minimizes the performance index   $J_{2,1}\left(\boldsymbol{U}\right)=\|\boldsymbol{U}\|_{2,1}=\sum_{i=0}^{N-1}\|\boldsymbol{u}(i)\|_2$ (for thrust vectoring) and $J_{1}\left(\boldsymbol{U}\right)=\left\|\boldsymbol{U}\right\|_{1}$ (for orthogonal vectoring) {subject to the following equality constraint:}
\begin{equation}
  \boldsymbol{\beta} + \boldsymbol{C}_N\boldsymbol{U} - \boldsymbol{x}_f= \boldsymbol{0}.
\end{equation}
\label{problem:performance_index}
\end{problem}
\begin{prop}
Problem \ref{problem:optimal_control_problem} and Problem \ref{problem:performance_index} are equivalent in the following sense: if~  $\boldsymbol{U}_{0: N-1}^{\star}:=\left\{\boldsymbol{u}^{\star}(k) \in \mathbb{R}^{3}: k\in[0, N-1]_d\right\}$ is a control sequence that solves Problem~1, then the corresponding vector $\boldsymbol{U}^{\star}=[\boldsymbol{u}^{\star}(0)^\mathrm{T}\; \dots \; \boldsymbol{u}^{\star}(N-1)^\mathrm{T}]^\mathrm{T}$ solves Problem~2, and vice versa.
\end{prop}

\begin{proof}
{For the thrust vectoring case, in view of  Eqn.~\eqref{eqn:performance_index_l2/l1}, it follows that \begin{equation}
    \mathcal{J}_{2,1}\left(\boldsymbol{U}_{0: N-1}\right):=  \left\|\boldsymbol{U}_{0: N-1}\right\|_{\ell_{2}/\ell_1} = \sum_{i=0}^{N-1}\|\boldsymbol{u}(i)\|_2:=J_{2,1}\left(\boldsymbol{U}\right),
\end{equation}
where the second equality follows from the fact that 
the input sequence $\boldsymbol{U}_{0: N-1}:=\left\{\boldsymbol{u}(k) \in \mathbb{R}^{3}: k\in[0, N-1]_d\right\}$ and the stacked vector $\boldsymbol{U}=[\boldsymbol{u}(0)^\mathrm{T},\;\boldsymbol{u}(1)^\mathrm{T},\;\dots \;,\boldsymbol{u}(N-1)^\mathrm{T}]^\mathrm{T}$ are in one-to-one correspondence. Similarly, in view of  Eqn.~\eqref{eqn:performance_index_l1}, the performance index in Problem \ref{problem:optimal_control_problem} for the orthogonal vectoring case satisfies the following equation:
\begin{equation}
    \mathcal{J}_{1}\left(\boldsymbol{U}_{0: N-1}\right):=  \left\|\boldsymbol{U}_{0: N-1}\right\|_{\ell_1} = \sum_{i=0}^{N-1}\|\boldsymbol{u}(i)\|_1:=J_{1}\left(\boldsymbol{U}\right).
\end{equation}
Now using Eqn.~\eqref{eqn:terminal_constraint}, the terminal constraint $\boldsymbol{x}(N)=\boldsymbol{x}_f$ in Problem \ref{problem:optimal_control_problem} can be written as the following linear constraint:
\begin{align}
   \boldsymbol{\beta} + \boldsymbol{C}_N\boldsymbol{U}- \boldsymbol{x}_f= \boldsymbol{0},
\end{align}
where in the last equation the initial condition $\boldsymbol{x}_0=\boldsymbol{x}(0)$ has been used.
One concludes that Problem \ref{problem:optimal_control_problem} and Problem \ref{problem:performance_index} are equivalent to each other.}
\end{proof}

\section{Proposed Solution Approach}\label{sec:algorithms}
In this section, two IRLS algorithms are proposed to generate sparse control sequences which minimize the net fuel consumption in the case of thrust vectoring and orthogonal vectoring respectively.

{The IRLS algorithm is used to find the minimum $\ell_1$ norm solution to an under-determined system of linear equations $\boldsymbol{y}= C \boldsymbol{x}$ where $C$ is an $M\times N$-dimensional matrix with $M<N$. A solution to the latter system will necessarily lie on a $(N-M)$ dimensional hyperplane. If $M$ is significantly smaller than $N$, then the solution to the under-determined system can admit a sparse representation. When the matrix $C$ enjoys the so-called restricted isometry property \cite{candes2005decoding_irls}, then the minimum $\ell_1$ norm solution is guaranteed to be a sparse vector. Powerful linear programming techniques can be utilized to find the minimum $\ell_1$ norm solution but their implementation requires specialized numerical techniques (e.g., interior point methods). In this work, we employ a much simpler and easily implementable approach called Iteratively Reweighted Least Squares (IRLS) algorithm. 
By using IRLS algorithm, one can find the minimum $\ell_1$ norm solution of an under-determined system of linear equations \cite{sparse_bakolas2019computation} or the minimum $\ell_2/\ell_1$ norm solution \cite{block_wang2013recovery} by finding the limit of a sequence of the minimum weighted $\ell_2$ norm (the norm is weighted by a positive definite weighting matrix) solution of the same under-determined linear system. The IRLS algorithm updates the weighting matrices at each iteration in such a way that it is ensured that the limit of the sequence corresponds to the minimum $\ell_1$ or $\ell_2/\ell_1$ norm solution. 
A detailed analysis of the IRLS algorithm can be found in \cite{robust_efficient_irls_2010iteratively}.

}

\subsection{Thrust vectoring}
In this section, Problem \ref{problem:performance_index} is addressed for the case of thrust vectoring. The main steps of the IRLS algorithm, which will generate control sequences that minimizes the $\ell_2/\ell_1$ control norm are described next (Algorithm \ref{algo:irls_sparse_l2/l1}). The presentation follows closely to \cite{sparse_bakolas2019computation,block_wang2013recovery}.

\subsubsection{IRLS algorithm for $\ell_2/\ell_1$ optimization}

\begin{algorithm}[H]
\caption{ IRLS algorithm for solving $\ell_2/\ell_1$ optimization problem}
\begin{algorithmic}[1]
\State $\boldsymbol{w}^{[0]}(i)=1\;\forall \;i\in [1,Nm]$
\State $\epsilon^{[0]}=1$
\For{$j=0\;\text{to}\;j_{\text{max}}$}
\For{$k=0\dots N-1$}
\State $\mathbf{W}^{[j]}(k)=\operatorname{diag}\left(\boldsymbol{w}^{[j]}{(k m+1)}, \ldots,\boldsymbol{w}^{[j]}{(k m+m)}\right)$
\EndFor
\State $\boldsymbol{\mathcal{W}}^{[j]}=\operatorname{bdiag}\left(\mathbf{W}^{[j]}(0), \ldots, \mathbf{W}^{[j]}(N-1)\right)$
\State $\boldsymbol{u}^{[j+1]}=\left(\boldsymbol{\mathcal{W}}^{[j]}\right)^{-1}\left(\boldsymbol{C}_{N}^{\mathrm{T}}\left(\boldsymbol{\mathcal{W}}^{[j]}\right)^{-1}\right)^{\mathrm{T}}\left(\boldsymbol{C}_{N}^{\mathrm{T}}\left(W^{[j]}\right)^{-1}+\tau \mathbf{I}\right)^{-1}\left(\boldsymbol{C}_{N}^{\mathrm{T}}\left(\boldsymbol{\mathcal{W}}^{[j]}\right)^{-1}\right)^{\mathrm{T}} \boldsymbol{\beta}$
\State $ \varepsilon^{[j+1]}=\min \left\{\varepsilon^{[j]}, \text{max}(\boldsymbol{u}^{[j+1]})\right\}$
\For{$\ell=1,\dots Nm$}
\State $\boldsymbol{w}^{[j+1]}{(\ell)}=\left(\left({u}^{[j+1]}{(\ell)}\right)^{2}+\left(\varepsilon^{[j+1]}\right)^{2}\right)^{-1/4}$
\EndFor
\EndFor
\end{algorithmic}
\label{algo:irls_sparse_l2/l1}
\end{algorithm}

First, the input parameters $\boldsymbol{w}^{[0]}(k)$ for all $k\in[1,Nm]$ and  $\varepsilon^{[0]}$ are initialized to 1 and $j$ is initially set to zero where $m$ is the dimension of the control input. In this case $m=3$. For a particular value of $k\in[0,N-1]$, the weight matrix  $\mathbf{W}^{[j]}(k)$ is defined as follows
\begin{align}
  \mathbf{W}^{[j]}(k)=\operatorname{diag}\left(\boldsymbol{w}^{[j]}{(k m+1)}, \ldots,\boldsymbol{w}^{[j]}{(k m+m)}\right) 
  \label{eqn:weight_matrix_l2}
\end{align}
for $k\in[0,N-1]$. In addition $\boldsymbol{\mathcal{W}}^{[j]} $ is defined as follows
\begin{align}
  \boldsymbol{\mathcal{W}}^{[j]}=\operatorname{bdiag}\left(\mathbf{W}^{[j]}(0), \ldots, \mathbf{W}^{[j]}(N-1)\right)  
  \label{eqn:weight_matirx_big_l2}
\end{align}
where the matrices $\mathbf{W}^{[j]}(k)$ for $k\in[0,N-1]$ and $\boldsymbol{\mathcal{W}}^{[j]}$ are positive definite (and thus non-singular) provided that $\boldsymbol{w}^{[j]}(k)\geq0$ for all $k\in[1,Nm]$. 

Then, the control sequence $\boldsymbol{u}^{[j+1]}$ is given by
\begin{align}
  \boldsymbol{u}^{[j+1]}=&\left(\boldsymbol{\mathcal{W}}^{[j]}\right)^{-1}\left(\boldsymbol{C}_{N}^{\mathrm{T}}\left(\boldsymbol{\mathcal{W}}^{[j]}\right)^{-1}\right)^{\mathrm{T}}\left(\boldsymbol{C}_{N}^{\mathrm{T}}\left(\boldsymbol{\mathcal{W}}^{[j]}\right)^{-1}+ \mathbf{I}\right)^{-1}\left(\boldsymbol{C}_{N}^{\mathrm{T}}\left(\boldsymbol{\mathcal{W}}^{[j]}\right)^{-1}\right)^{\mathrm{T}} \boldsymbol{\beta},
  \label{eqn:control_vector_l2}
\end{align}
where $\boldsymbol{C}_N$ and $\boldsymbol{\beta}$ are given by Eqns.~\eqref{eqn:cn} and~\eqref{eqn:beta} respectively.

Now $ \varepsilon^{[j+1]}$ is set equal to $\min \left\{\varepsilon^{[j]}, \text{max}(\boldsymbol{u}^{[j+1]})\right\},$ where $\text{max}(\boldsymbol{u}^{[j+1]})$ denotes the maximum element of the vector $\boldsymbol{u}^{[j+1]}$. Then $\boldsymbol{w}^{[j+1]}{(\ell)}$ is updated again as follows
\begin{align}
  \boldsymbol{w}^{[j+1]}{(\ell)}=\left(\left({u}^{[j+1]}{(\ell)}\right)^{2}+\left(\varepsilon^{[j+1]}\right)^{2}\right)^{-1 / 4}
\end{align}
for all $\ell\in[1,Nm]$, where ${u}^{[j+1]}{(\ell)}$ are the elements in the vector $\boldsymbol{u}^{[j+1]}$ from Eqn. \eqref{eqn:control_vector_l2}.
Now $j$ is set to $j+1$. Next, the updated $\boldsymbol{w}^{[j]}{(\ell)}$ is used to update the matrix $\mathbf{W}^{[j]}(k)$ and then update matrices $\boldsymbol{\mathcal{W}}^{[j]}$ and $\boldsymbol{u}^{[j+1]}$ given by Eqns. \eqref{eqn:weight_matirx_big_l2} and \eqref{eqn:control_vector_l2}. This process is repeated until the control sequence $\boldsymbol{u}$ converges to the optimal control sequence $\boldsymbol{u}^*_\text{IRLS}$. If $j $ is less than or equal to $ j_{\max }$ and $\varepsilon^{[j]} \in[0, \bar{\varepsilon}],$ then Algorithm \ref{algo:irls_sparse_l2/l1} is terminated successfully. Else if $\varepsilon^{[j]} \notin[0, \bar{\varepsilon}]$, two cases arise. First, if $j<j_{\max },$ then go to Eqn. \eqref{eqn:weight_matrix_l2} and second if $j=j_{\max }$, the algorithm failed to converge. Hence, it is suggested to set a larger $j_{\max}$ to increase the chances of success.

\subsection{Orthogonal vectoring}
In this section, an IRLS algorithm is proposed to address Problem \ref{problem:performance_index} for the case of orthogonal vectoring. Note that in this case minimization of $\sum_{i=0}^{N-1}\|\boldsymbol{u}(i)\|_1$ in view of \eqref{eqn:performance_index_l1} is equivalent to minimizing $\|\boldsymbol{U}\|_1$ as 
\begin{align}
    \sum_{i=0}^{N-1}\|\boldsymbol{u}(i)\|_1=\|\boldsymbol{U}\|_1
\end{align}
where $\boldsymbol{U}=[\boldsymbol{u}(0)^\mathrm{T},\;\boldsymbol{u}(1)^\mathrm{T},\; \dots, \;\boldsymbol{u}(N-1)^\mathrm{T}]^{\mathrm{T}}$.

Next, the main steps of the proposed algorithm (Algorithm \ref{algo:irls_sparse_l1}) are presented. The exposition follows closely to \cite{sparse_bakolas2019computation}.
First, the input parameters $\boldsymbol{w}^{[0]}(k)$ for all $k\in [1,Nm]$ and  $\varepsilon^{[0]}$ are initialized to 1 and $j$ is set equal to zero.
 For a particular value of $k\in[0,N-1]$, the weight matrices  $\mathbf{W}^{[j]}(k)$ and $\mathcal{W}^{[j]}$ are defined as in Eqns.\eqref{eqn:weight_matrix_l2} and \eqref{eqn:weight_matirx_big_l2} respectively.
Then, $\boldsymbol{u}^{[j+1]}$ is updated as follows
\begin{align}
  \boldsymbol{u}^{[j+1]}=\left(\boldsymbol{\mathcal{W}}^{[j]}\right)^{-1} \boldsymbol{{C}}_{N}^{\mathrm{T}} \mathcal{G}^{[j]}(N)^{-1} \boldsymbol{\beta} 
  \label{eqn:control_vector_l1}
\end{align}
where $\mathcal{G}^{[j]}(k)$ satisfies the following recursive Lyapunov (matrix) equation:
\begin{align}
  \mathcal{G}^{[j]}(k+1)={A}(k) \mathcal{G}^{[j]}(k) {A}(k)^{\mathrm{T}}+{B}(k)\left(\mathbf{W}^{[j]}(N-1-k)\right)^{-1}{B}(k)^{\mathrm{T}}  
\end{align}
for $k \in[0, N-1]$ and  $\mathcal{G}^{[j]}(0)={B}(0)\left(\mathbf{W}^{[j]}(N-1)\right)^{-1} {B}(0)^{\mathrm{T}}$.

\subsubsection{IRLS algorithm for $\ell_1$ optimization}
\begin{algorithm}[H]
\caption{IRLS algorithm for solving $\ell_1$ optimization problem}
\begin{algorithmic}[1]
\State $\boldsymbol{w}^{[0]}(i)=1\;\forall \;i\in[1,Nm]$
\State $\epsilon^{[0]}=1$
\For{$j=0\;\text{to}\;j_\text{max}$}
\For{$k=0,\dots N-1$}
\State $\mathbf{W}^{[j]}(k)=\operatorname{diag}\left(\boldsymbol{w}^{[j]}{(k m+1)}, \ldots,\boldsymbol{w}^{[j]}{(k m+m)}\right)$
\EndFor
\State $\boldsymbol{\mathcal{W}}^{[j]}=\operatorname{bdiag}\left(\mathbf{W}^{[j]}(0), \ldots, \mathbf{W}^{[j]}(N-1)\right) $
\State $\mathcal{G}^{[j]}(0)={B}(0)\left(\mathbf{W}^{[j]}(N-1)\right)^{-1} {B}(0)^{\mathrm{T}}$
\For{$i=0\;\text{to}\;N-1$}
\State $\mathcal{G}^{[j]}(i+1)={A}(i) \mathcal{G}^{[j]}(i) {A}(i)^{\mathrm{T}}+{B}(i)\left(\mathbf{W}^{[j]}(N-1-i)\right)^{-1} {B}(i)^{\mathrm{T}}$
\EndFor
\State $\boldsymbol{u}^{[j+1]}=\left(\boldsymbol{\mathcal{W}}^{[j]}\right)^{-1} \boldsymbol{C}_{N}^{\mathrm{T}} \boldsymbol{\mathcal { G }}^{[j]}(N)^{-1} \boldsymbol{\beta}$
\For{$\ell=1,\dots Nm$}
\State $\left(\boldsymbol{w}\right)^{[j+1]}{(\ell)}=\left(\left({u}^{[j+1]}{(\ell)}\right)^{2}+\left(\varepsilon^{[j+1]}\right)^{2}\right)^{-1/2}$
\EndFor
\State $\varepsilon^{[j+1]}=\min \left\{\varepsilon^{[j]}, \text{max}(\boldsymbol{u}^{[j+1]})\right\}$ 
\EndFor
\end{algorithmic}
\label{algo:irls_sparse_l1}
\end{algorithm}

Now the value of $ \varepsilon^{[j+1]}$ is set equal to $\min \left\{\varepsilon^{[j]}, \text{max}(\boldsymbol{u}^{[j+1]})\right\},$ where $\text{max}(\boldsymbol{u}^{[j+1]})$ denotes the maximum element in the vector $\boldsymbol{u}^{[j+1]}$. Then, $\boldsymbol{w}^{[j+1]}{(\ell)}$ is updated as follows
\begin{align}
  \boldsymbol{w}^{[j+1]}{(\ell)}=\left(\left({u}^{[j+1]}{(\ell)}\right)^{2}+\left(\varepsilon^{[j+1]}\right)^{2}\right)^{-1 / 2}
\end{align}
for all $\ell\in[1,Nm]$.  $\boldsymbol{u}^{[j+1]}{(l)}$ are the elements of the vector $\boldsymbol{u}^{[j+1]}$ from Eqn. \eqref{eqn:control_vector_l1}.
Now $j$ is set to $j+1$. The updated $\boldsymbol{w}^{[j+1]}{(\ell)}$ is now used to update the matrix $\mathbf{W}^{[j]}(k)$ which is subsequently used to update matrices in Eqn. \eqref{eqn:weight_matirx_big_l2} and \eqref{eqn:control_vector_l1}. This process is repeated until the control sequence $\boldsymbol{u}^{[j+1]}$ converges to the optimal control sequence $\boldsymbol{u}^*_\text{IRLS}$. If $j \leq j_{\max }$ and $\varepsilon^{[j]} \in[0, \bar{\varepsilon}]$, then report success and stop. If $\varepsilon^{[j]} \notin[0, \bar{\varepsilon}],$ then two cases arise. First, if $j<j_{\max },$ go to Eqn. \eqref{eqn:weight_matrix_l2}  and if $j=j_{\max }$ report failure.


\section{Simulation Results}\label{sec:simulation_results}
In this section, numerical simulations are performed to demonstrate the effectiveness of the proposed modified versions of the IRLS algorithm to solve the minimum fuel problem for both the orthogonal vectoring and thrust vectoring cases.

\subsection{Out-of-plane maneuvers for a Geostationary Transfer Orbit (GTO) Mission}
Consider a target passive spacecraft in the elliptical orbit with eccentricity $e=0.73074$ with initial and final true anomaly equal to $0.1\pi$ and $5.2 \mathrm{rad}$ respectively. The initial position and velocity of the active spacecraft in the $Y$ direction with respect to the LVLH frame is $10000\mathrm{m}$ and $-3\mathrm{m/s}$ respectively. The task is to generate optimal control sequences to take the active spacecraft from the initial true anomaly and initial state to the final true anomaly and final state given by $\boldsymbol{x}_f=[0\mathrm{m}\;\;0\mathrm{m}\;\;0\mathrm{m}\;\;0\mathrm{m/s}\;\;0\mathrm{m/s}\;\;0\mathrm{m/s}]^\mathrm{T}$. The orbital rendezvous parameters are given in Table \ref{tab:gto_mission_parameters} and taken from \cite{zhou2011lyapunov_gto}. By numerical simulations, the state trajectories and control signals of the discrete linear system described by Eqn.~\eqref{eqn:discrete_linear_sytem_final} are given in Fig. \ref{fig:IRLS_algorithm_gto}. As seen from the numerical simulations, the IRLS algorithm generates optimal control sequences which takes the active spacecraft from its initial to its final states. It is observed that the minimum $\ell_1$ norm for fuel consumption is around $2\%$ greater than the algorithm proposed in \cite{arzelier2016linearized_compare}. Hence the IRLS algorithm is able to generate control sequences whose $\ell_1$ norm is close to the optimal ones. The state trajectories and control inputs for $N=200$, $N=300$ and $N=600$ are shown in Figures \ref{fig:IRLS_algorithm_gto_n_200}-\ref{fig:IRLS_algorithm_gto}. {Through simulations, it is observed that the value of $N$ must be atleast $200$ to get the desired accuracy for GTO mission}. Further it is also observed that both the states and control sequences are not sensitive to the choice of the sampling period for $N\geq 200$.


The rendezvous parameters for GTO mission is given in Table \ref{tab:gto_mission_parameters} and are taken from \cite{zhou2011lyapunov_gto}.


\begin{table}
\caption{Parameters for GTO mission}
\label{tab:gto_mission_parameters}
    \centering
    \begin{tabular}{ll}
       \hline \hline Semi-major axis & ${a}=24616 \mathrm{km}$ \\
 Eccentricity & ${e}=0.73074$ \\
 Initial anomaly & $\nu_{0}=0.1 \pi \;\mathrm{rad}$ \\
 Initial position & $10000 \mathrm{m}$ \\
 Initial velocity & $-3\mathrm{m} / \mathrm{s}$ \\
 Final anomaly & $\nu_{f}=5.2 \mathrm{rad}$ \\
 Final state vector & $X_{f}^{T}=[0\mathrm{m}\; 0\mathrm{m}/\mathrm{s}]$ \\
\hline \hline
    \end{tabular}
     \end{table}
    \begin{table}
    \caption{Performance of Control Sequences for GTO mission}
    \centering
\begin{tabular}{ll}
\hline \hline Control Algorithm & $\left\|U_{0: N-1}\right\|_{\ell_{1}}$ \\
\hline Minimum $\ell_{1}$ norm (IRLS algorithm) & 6.4211\\ 
 Arzelier et. al method \cite{arzelier2016linearized_compare} & 6.2725 \\ \hline \hline
\end{tabular}
    \label{tab:my_label}
\end{table}


\begin{figure}[h]
\centering
\includegraphics[width=8cm]{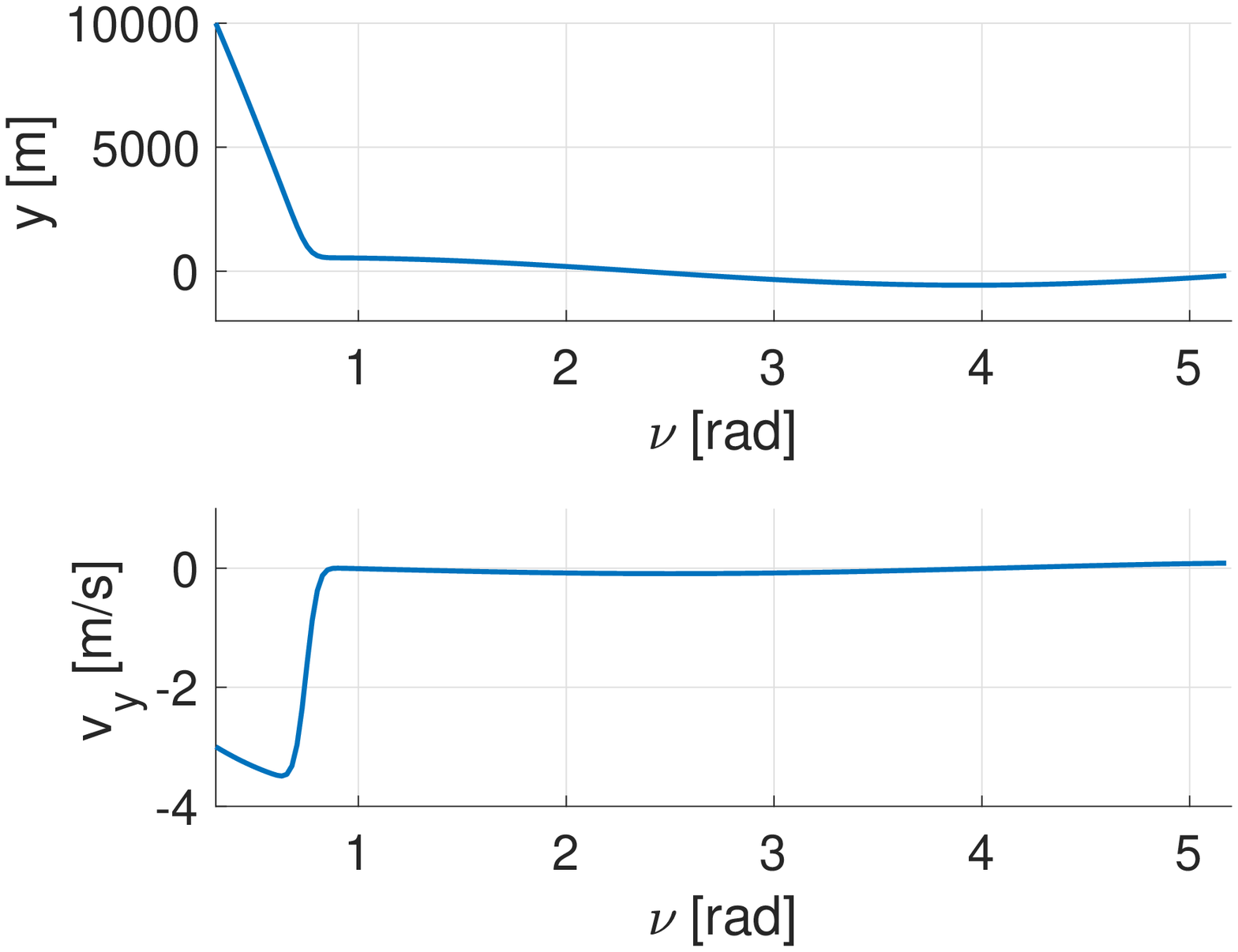}
\includegraphics[width=8cm]{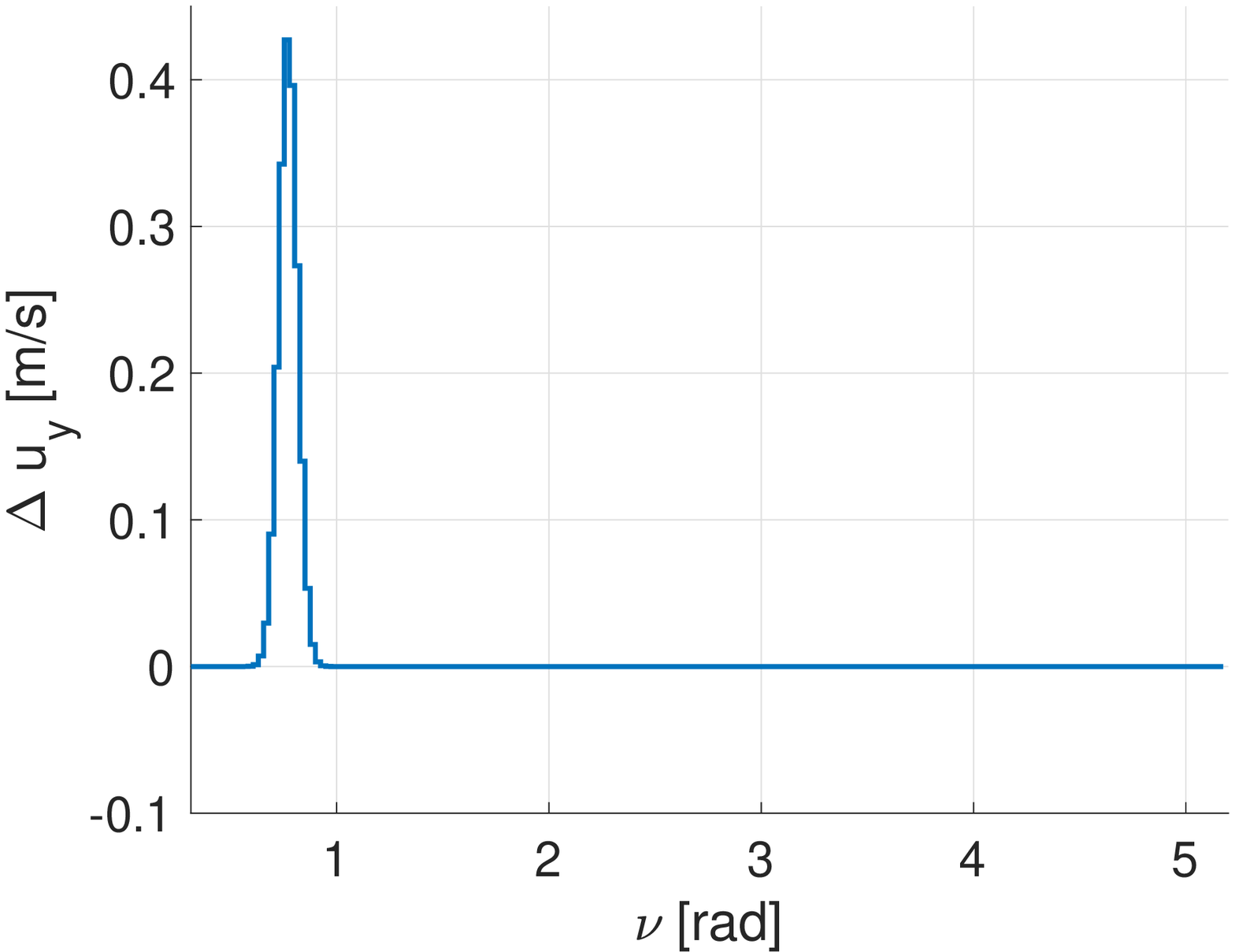}
\caption{IRLS algorithm (Algorithm \ref{algo:irls_sparse_l1}) for solving $\ell_1$ optimization problem in case of out of plane GTO mission (N=200)}
\label{fig:IRLS_algorithm_gto_n_200}
\end{figure}

\begin{figure}[h]
\centering
\includegraphics[width=8cm]{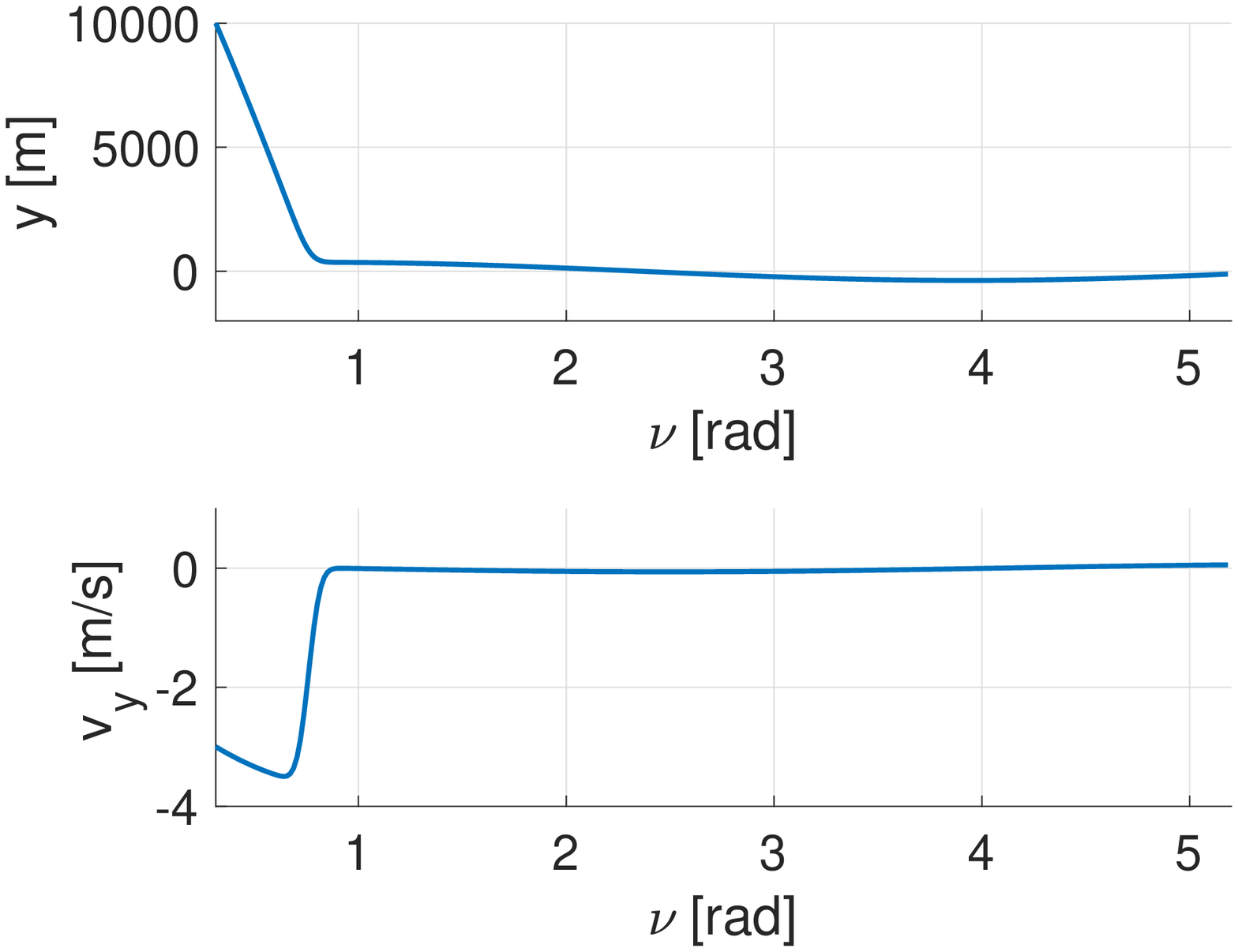}
\includegraphics[width=8cm]{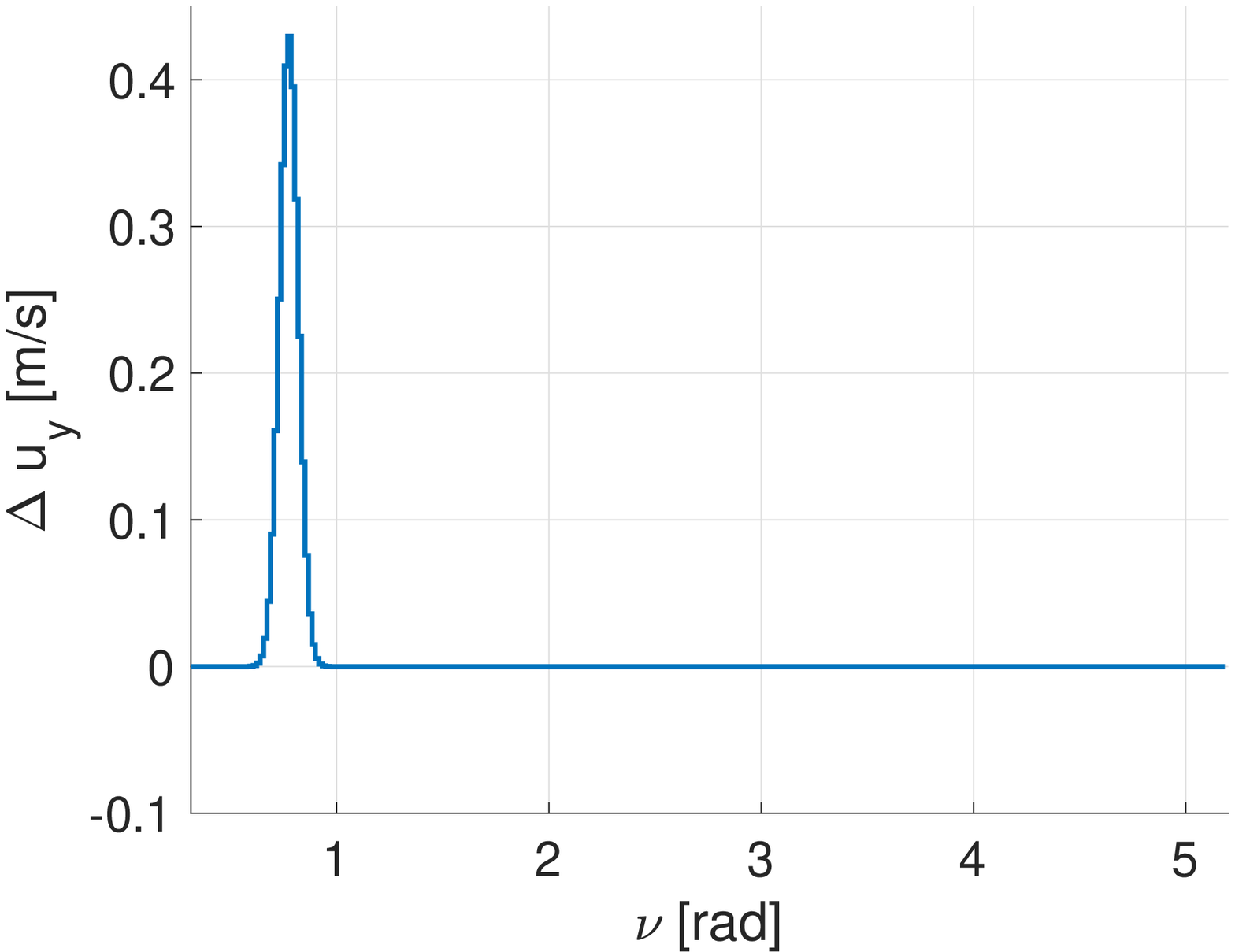}
\caption{IRLS algorithm (Algorithm \ref{algo:irls_sparse_l1}) for solving $\ell_1$ optimization problem in case of out-of-plane GTO mission (N=300)}
\label{fig:IRLS_algorithm_gto_n_300}
\end{figure}

\begin{figure}[h]
\centering
\includegraphics[width=8cm]{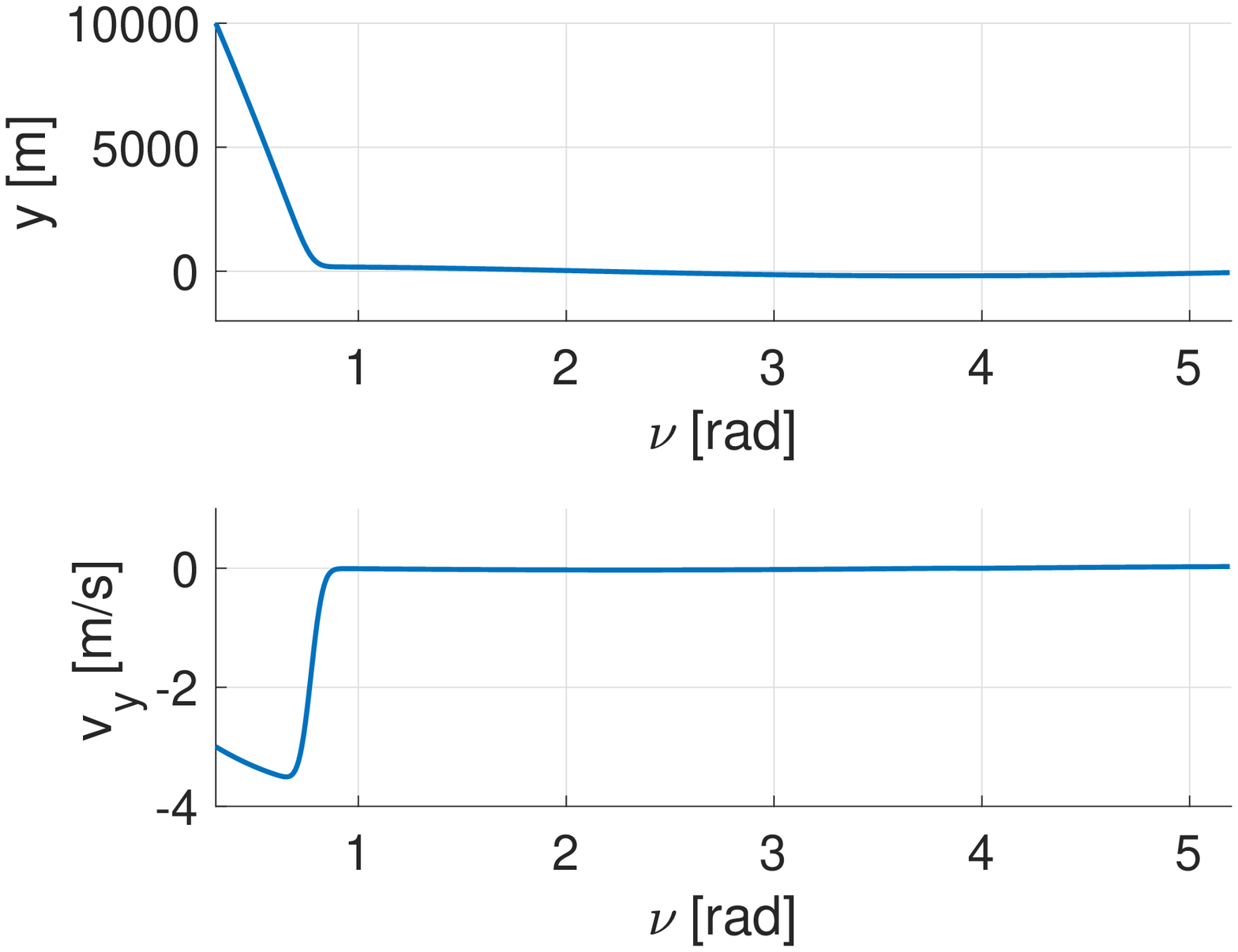}
\includegraphics[width=8cm]{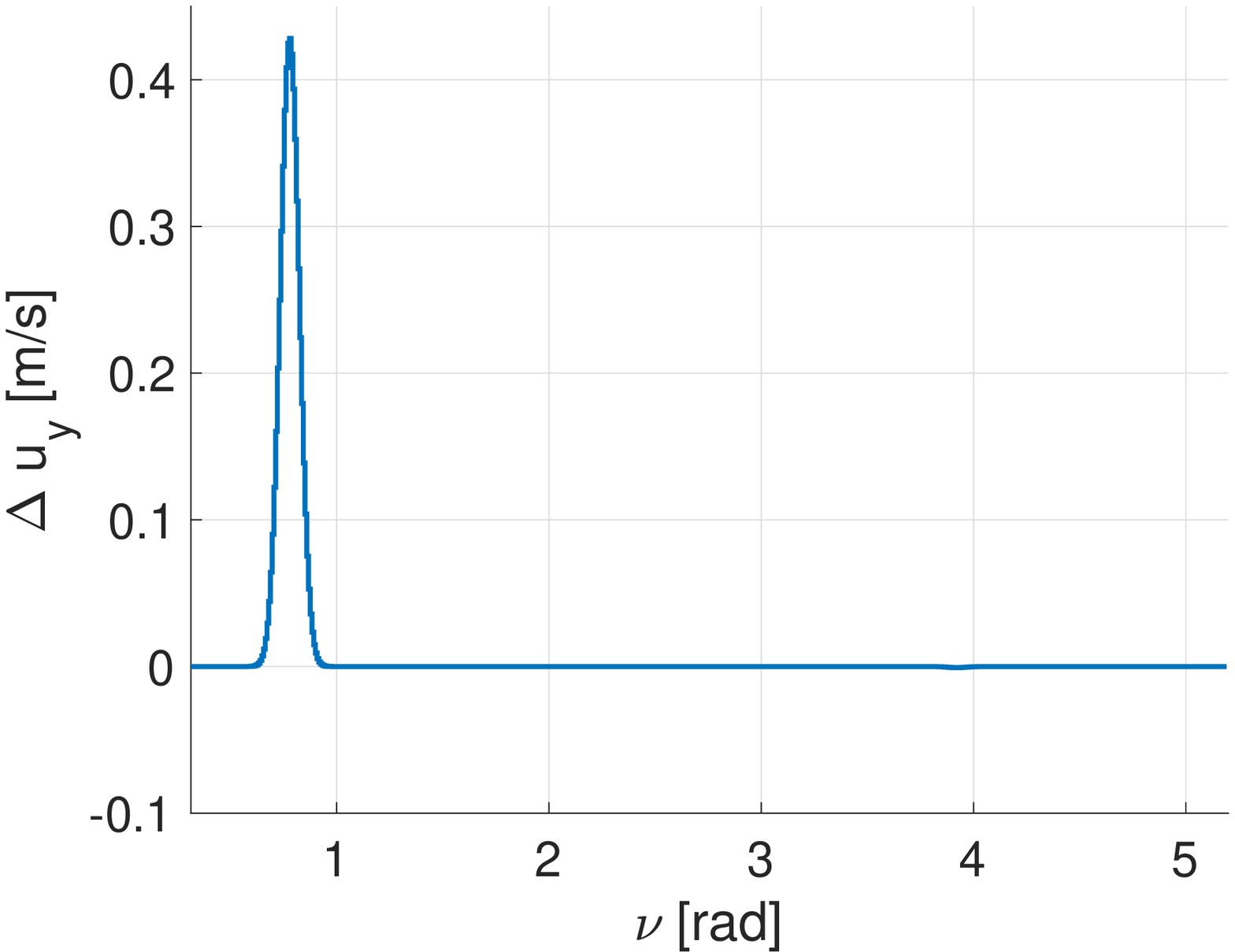}
\caption{IRLS algorithm (Algorithm \ref{algo:irls_sparse_l1}) for solving $\ell_1$ optimization problem in case of out-of-plane GTO mission (N=600)}
\label{fig:IRLS_algorithm_gto}
\end{figure}
\begin{figure}[h]
\centering
\includegraphics[width=8cm]{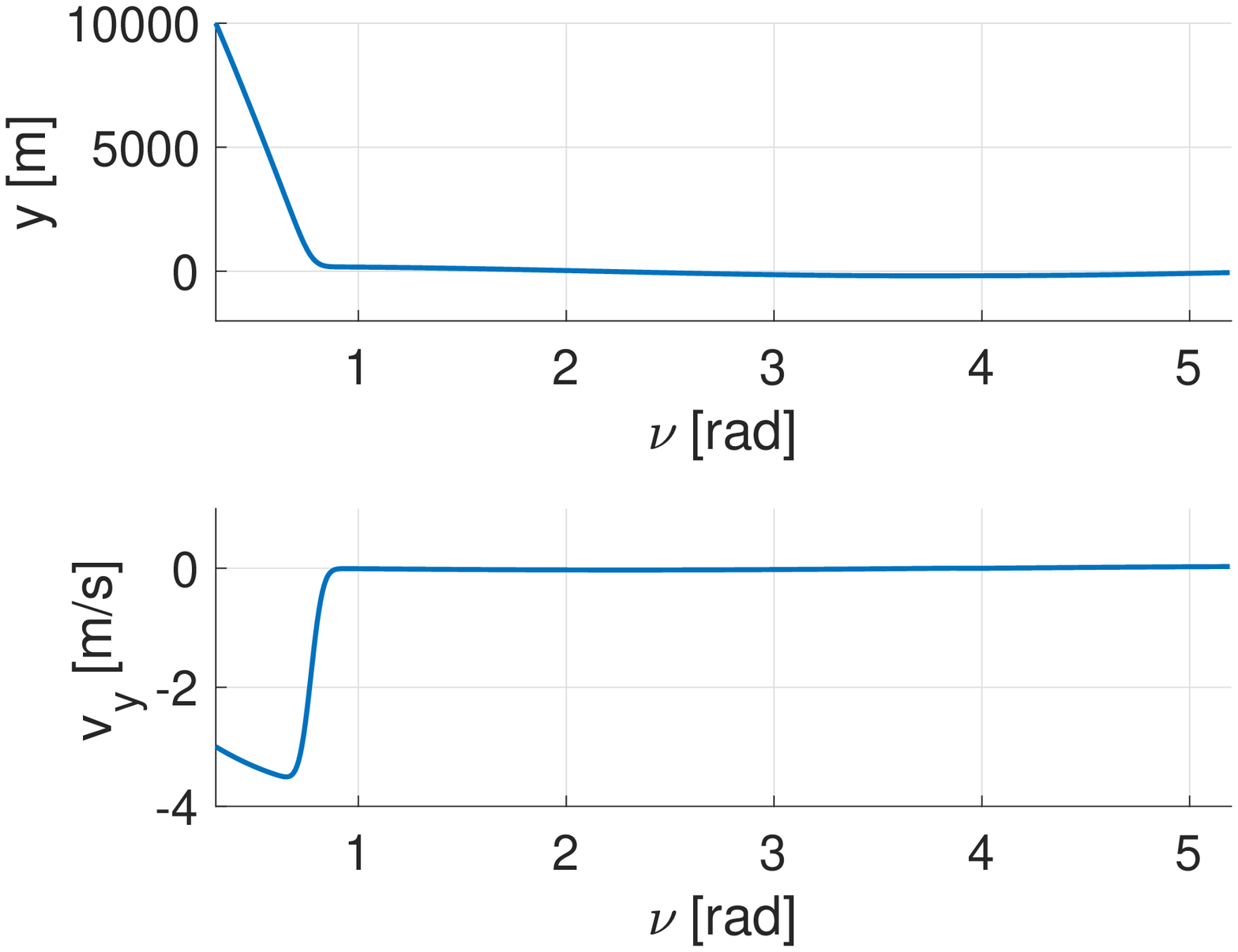}
\caption{Evolution of states when control inputs generated from Algorithm \ref{algo:irls_sparse_l1} are applied to continuous linearized rendezvous equation given in Eqn. \eqref{eqn:continous_state_space_model} (N=600)}
\label{fig:IRLS_algorithm_gto_nonlinear}
\end{figure}

\subsection{Coplanar maneuvers for Automated Vehicle Transfer (ATV) Mission}
An Automated Vehicle Transfer (ATV) mission is considered where the in-plane motion is as in \cite{arzelier2016linearized_compare}. For the in-plane rendezvous, two different examples are studied, one is a single gimbaled thruster (thrust vectoring) in which the fuel consumption is determined by the $\ell_2/\ell_1$ norm of the control sequence. Secondly, the ungimbaled thrusters are considered whose fuel consumption is determined by the $\ell_1$ norm of the control sequence (orthogonal vectoring).
It was observed that the value of the minimum $\ell_2/\ell_1$ and $\ell_1$ norm using the  method proposed herein was close (within 2$\%$ in case of $\ell_1$ and $\ell_2/\ell_1$ norm) to the optimal value using method proposed in \cite{arzelier2016linearized_compare}. The parameters for the ATV mission are taken from \cite{arzelier2016linearized_compare}.

\subsubsection{$\ell_1$ norm minimization (orthogonal vectoring)}
Since the active spacecraft is moving only in the $X-Z$ plane, the control input is given by $\boldsymbol{u}=[u_x\;0\;u_z]^\mathrm{T}$. The optimal control sequence and their locations using IRLS algorithm (Algorithm \ref{algo:irls_sparse_l1}) for minimizing $\ell_1$ norm is given as follows:

Control input in $X$ direction ($u_x$):
\begin{align}
    \text{Optimal control sequence}=\{-8.211,~~-0.889,~~1.752,~~0.214\}\mathrm{m/s}\nonumber\\
    \text{Location of corresponding true anomalies}=\{0,~~3.600,~~7.856,~~8.019\}\mathrm{rad}\nonumber 
\end{align}
The control input in the $Z$ direction is found to be equal to $0$ for all true anomalies in range $[\nu_o,\nu_f]$, i.e., $u_z \equiv 0$. The evolution of the states and the control inputs ( minimum $\ell_1$-norm soluion) for the ATV mission is shown in Fig. \ref{fig:IRLS_algorithm_l1_atv}. The small bumps in the control inputs shown in Figure \ref{fig:IRLS_algorithm_l1_atv} correspond to small amplitude impulse-like corrections computed by the numerical implementation of Algorithm \ref{algo:irls_sparse_l1}. The optimal trajectory in the $X-Z$ plane (in-plane motion) is illustrated in Fig. \ref{fig:traj_x_z_atv_mission} and the corresponding value of $\ell_1$-norm is given in Table \ref{tab:atv_l1}.

\subsubsection{$\ell_2/\ell_1$ norm minimization (thrust vectoring)}
The active spacecraft is moving in the $X-Z$ plane only. Hence, the control input is given by $\boldsymbol{u}=[u_x\;0\;u_z]^\mathrm{T}$. The optimal control sequence and their locations using IRLS algorithm (Algorithm \ref{algo:irls_sparse_l2/l1}) for minimizing $\ell_2/\ell_1$ norm is given as follows:

Control input in $X$ direction ($u_x$):
\begin{align}
    \text{Optimal control sequence}=\{-8.117~-0.132~~-0.848~~-0.001~~0.140~~1.505~~0.318\}\mathrm{m/s}\nonumber\\
    \text{Location of corresponding true anomalies}=\{0~~3.437~~3.600~~3.764~~7.856~~8.019~~8.183\}\mathrm{rad}\nonumber 
\end{align}
The control input in the $Z$ direction is found to be equal $0$ for all true anomalies in range $[\nu_o,\nu_f]$, i.e., $u_z \equiv 0$. The corresponding state trajectories and control inputs are shown in Fig. \ref{fig:IRLS_algorithm_l2/l1_atv}. The following parameters for Automated Transfer Orbit (ATV) mission are taken from \cite{labourdette2008atv}. The small bumps in the control inputs shown in Figure~ \ref{fig:IRLS_algorithm_l2/l1_atv} correspond to small amplitude impulse-like corrections computed by the numerical implementation of Algorithm \ref{algo:irls_sparse_l2/l1}.

\begin{table}
\caption{Parameters for ATV mission}
\label{tab:atv_mission_parameters}
\centering
\begin{subtable}[t]{0.51\textwidth}
    \centering
    \begin{tabular}{ll}
       \hline \hline Semi-major axis & ${a}=6763 \mathrm{km}$ \\
 Eccentricity & ${e}=0.0052$ \\
 Initial anomaly & $\nu_{0}=0 \;\mathrm{rad} .$ \\
 Initial state vector $X_{0}^{T}$ & {$[-30\mathrm{km}\;0.5\mathrm{km}\;8.5140\mathrm{m}/\mathrm{s}\;0\mathrm{m}/\mathrm{s}]$} \\
 Final anomaly & $\nu_{f}=8.1831 \mathrm{rad}$ \\
 Final state vector $X_{f}^{T}$ & {$[-100\mathrm{m}\;0\mathrm{m}\;0\mathrm{m}/\mathrm{s}\;0\mathrm{m}/\mathrm{s}]$} \\
\hline \hline
    \end{tabular}
    \end{subtable}
    \end{table}
    
    \begin{table}
    \centering
\caption{Performance of Control Sequences}
\label{tab:atv_l1}
    \centering
\begin{tabular}{ll}
\hline \hline Control Algorithm & $\left\|U_{0: N-1}\right\|_{\ell_{1}}$ \\
\hline Minimum $\ell_1$ norm (Algorithm \ref{algo:irls_sparse_l1}) & 11.0677\\
 Arzelier et. al method \cite{arzelier2016linearized_compare} & 10.8415 \\ \hline \hline
\end{tabular}
   \end{table}
   \begin{table}
    \caption{Performance of Control Sequences}
    \label{tab:atv_l2/l1}
    \centering
\begin{tabular}{ll}
\hline \hline Control Algorithm & $\left\|U_{0: N-1}\right\|_{\ell_{2}/\ell_1}$ \\
\hline Minimum $\ell_{2}/\ell_1$ norm (Algorithm \ref{algo:irls_sparse_l2/l1}) & 11.0623\\ 
 Arzelier et. al method \cite{arzelier2016linearized_compare} & 10.7989 \\ \hline \hline
\end{tabular}
\end{table}

\begin{figure}[h]
\centering
\includegraphics[width=8cm]{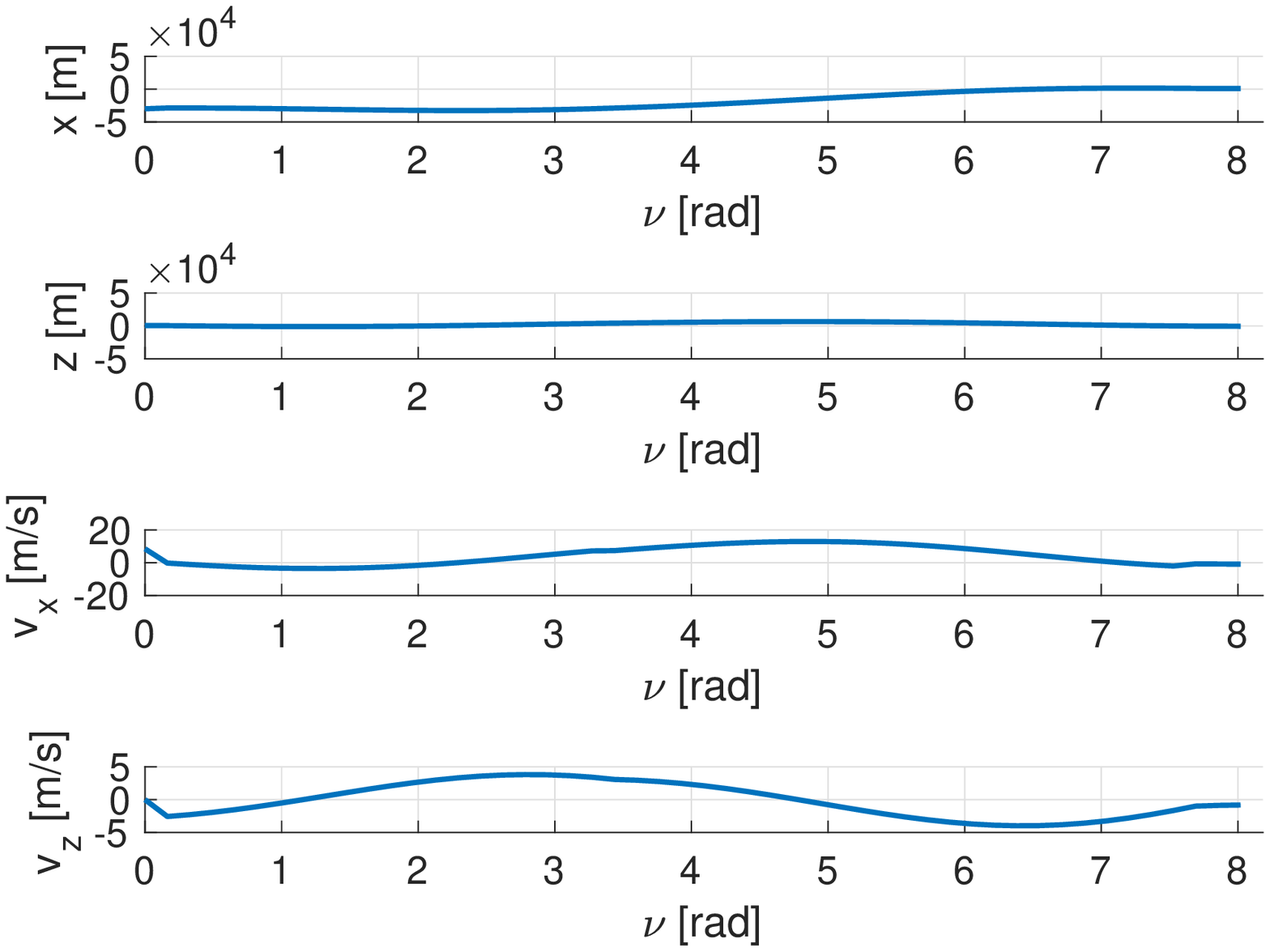}
\includegraphics[width=8cm]{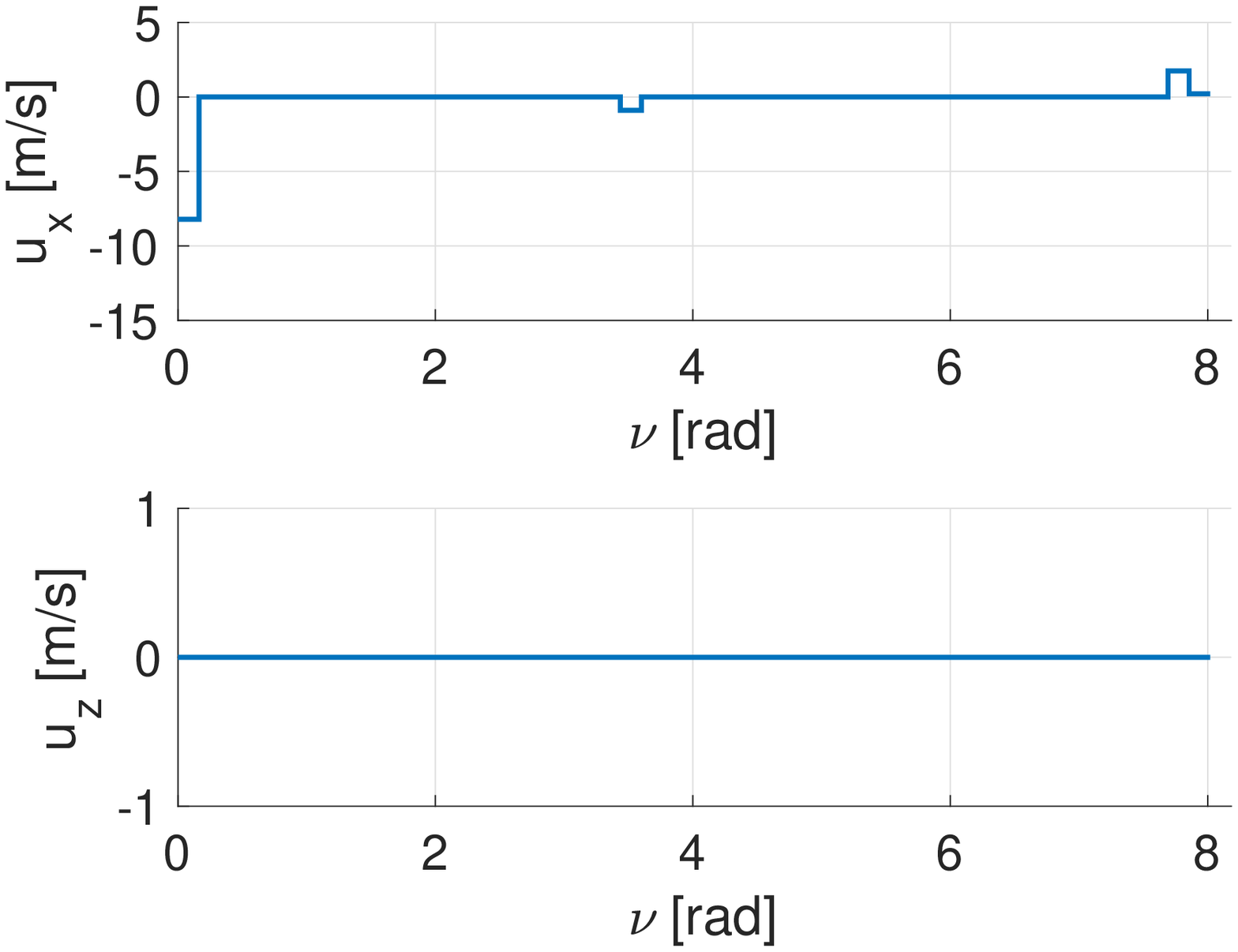}
\caption{IRLS algorithm for solving the $\ell_1$ optimization problem for in-plane ATV mission (N=50)}
\label{fig:IRLS_algorithm_l1_atv}
\end{figure}

\begin{figure}[H]
\centering
\includegraphics[width=13cm]{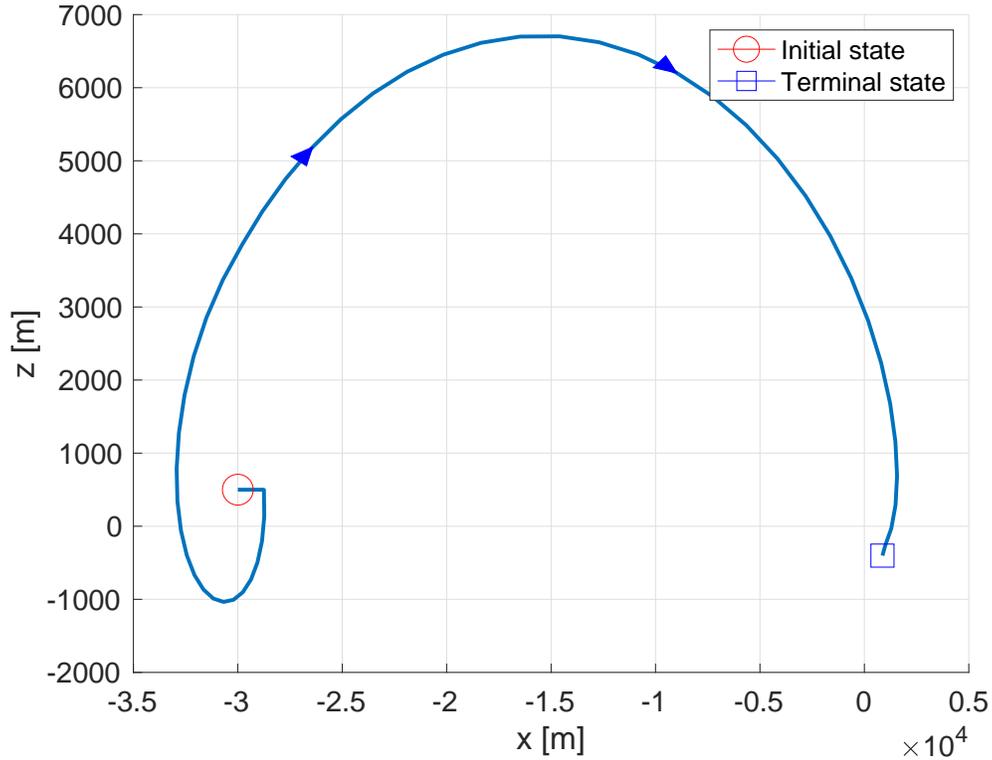}
\caption{Optimal $\ell_1$ norm trajectory in the $X-Z$ plane for the ATV mission}
\label{fig:traj_x_z_atv_mission}
\end{figure}

\begin{figure}[H]
\centering
\includegraphics[width=8cm]{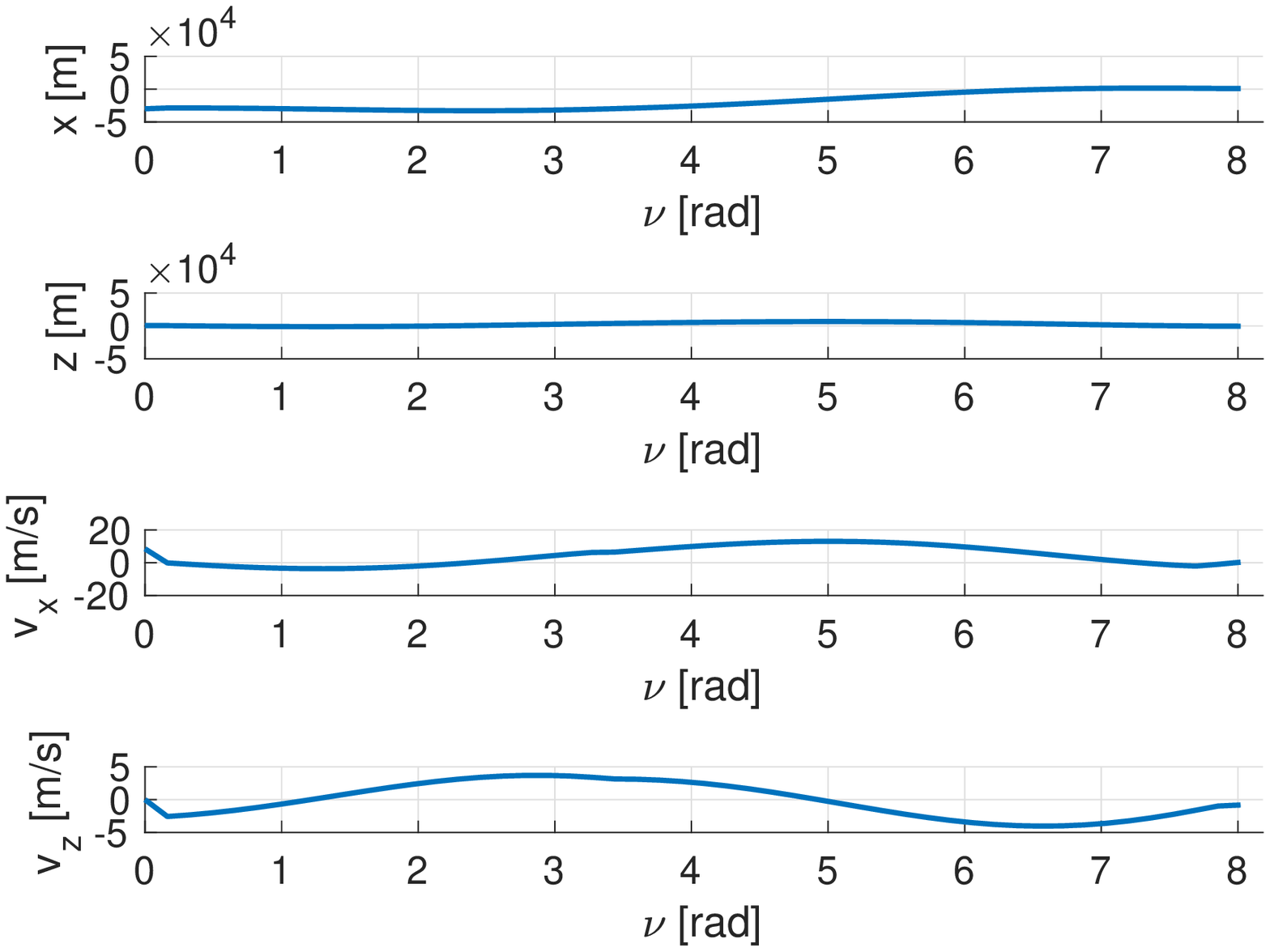}
\includegraphics[width=8cm]{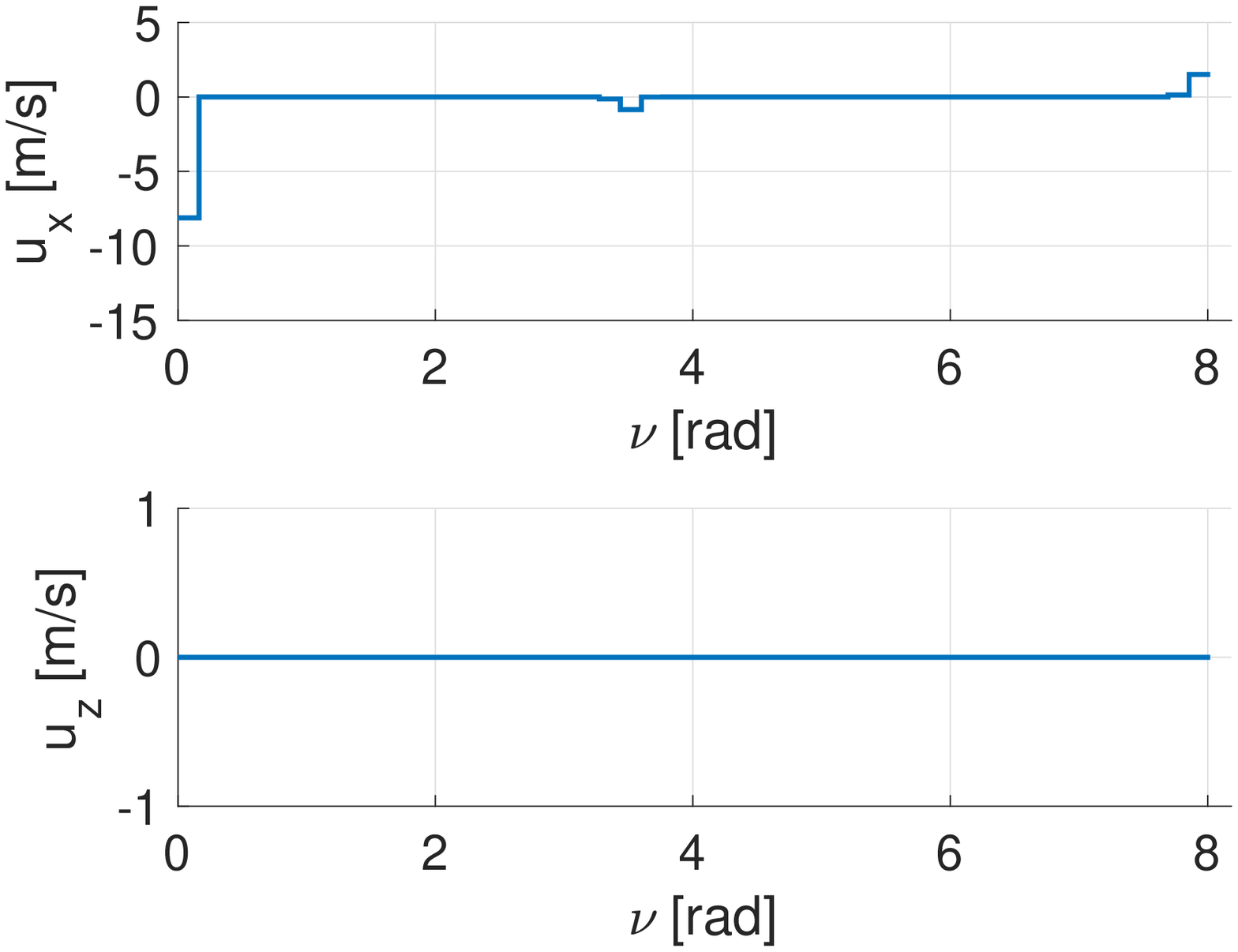}
\caption{IRLS algorithm for solving $\ell_2/\ell_1$ optimization problem for inplane ATV mission (N=50) }
\label{fig:IRLS_algorithm_l2/l1_atv}
\end{figure}


\section{Conclusion} \label{sec:conclusions}

In this paper, two iterative solution approaches are presented for the computation of approximate solutions to the minimum-fuel rendezvous problem assuming thrust vectoring and orthogonal vectoring control modules. The proposed techniques, which are both based on the Iteratively Reweighted Least Squares algorithm from compressive (or compressed) sensing, compute control sequences with minimum $\ell_2/\ell_1$ norm and $\ell_1$ norm for the thrust vectoring and orthogonal vectoring cases, respectively. In the proposed problem formulation, the dynamics of the relative spacecraft rendezvous is expressed in the Local-Vertical-Local-Horizontal frame in which true anomaly $\nu$ is taken to be the independent variable instead of time $t$. One of the main reasons for using true anomaly as {the independent variable has to do with the fact that it takes values in a compact interval which is significantly shorter than the corresponding time interval and thus it is more suitable for discretization purposes (e.g., less reliance of results on the sampling period) and control design (which requires the solution of smaller size optimization problems)}. Numerical simulations were performed to validate the proposed algorithms. The computed state trajectories converged to their desired final values and in addition, the values of the $\ell_2/\ell_1$ norm and $\ell_1$ norm of the corresponding control sequences were close to the optimal values. The main motivation for using an approach that is based on the Iteratively Reweighted Least Squares algorithm over other approaches proposed in the literature is mainly its simplicity, which allows its straightforward implementation even by non-experts, and the small computational cost and short execution time. It is argued that the proposed approach can be reliably executed onboard spacecraft with limited computational resources performing proximity operations. In future work, modifying the IRLS algorithm to incorporate control constraints for spacecraft rendezvous applications will be explored.
 
\bibliography{main}

\begin{thebibliography}{33}
\newcommand{\enquote}[1]{``#1''}
\providecommand{\natexlab}[1]{#1}
\providecommand{\url}[1]{\texttt{#1}}
\providecommand{\urlprefix}{URL }
\expandafter\ifx\csname urlstyle\endcsname\relax
  \providecommand{\doi}[1]{\discretionary{}{}{}https://doi.org/#1}\else
  \providecommand{\doi}[1]{\discretionary{}{}{}\urlstyle{rm}\url{https://doi.org/#1}}\fi

\bibitem[{Chamberlin and Rose(1964)}]{chamberlin1964gemini_1}
Chamberlin, J.~A., and Rose, J.~T., \enquote{Gemini rendezvous program,}
  \emph{Journal of Spacecraft and Rockets}, Vol.~1, No.~1, 1964, pp. 13--18.
\newblock \doi{10.2514/3.27585}.

\bibitem[{Burton and Hayes(1966)}]{burton1966gemini_2}
Burton, J., and Hayes, W., \enquote{Gemini rendezvous,} \emph{Journal of
  Spacecraft and Rockets}, Vol.~3, No.~1, 1966, pp. 145--147.
\newblock \doi{10.2514/6.1964-641}.

\bibitem[{Goodman(2012)}]{goodman2006history}
Goodman, J.~L., \enquote{History of space shuttle rendezvous and proximity
  operations,} \emph{Journal of Spacecraft and Rockets}, Vol.~43, No.~5, 2012,
  pp. 944--959.
\newblock \doi{10.2514/1.19653}.

\bibitem[{Bakolas(2019)}]{sparse_bakolas2019computation}
Bakolas, E., \enquote{On the Computation of Sparse Solutions to the
  Controllability Problem for Discrete-Time Linear Systems,} \emph{Journal of
  Optimization Theory and Applications}, Vol. 183, No.~1, 2019, pp. 292--316.
\newblock \doi{10.1007/s10957-019-01532-9}.

\bibitem[{Leomanni et~al.(2019)Leomanni, Bianchini, Garulli, Giannitrapani, and
  Quartullo}]{leomanni2019sum}
Leomanni, M., Bianchini, G., Garulli, A., Giannitrapani, A., and Quartullo, R.,
  \enquote{Sum-of-Norms Model Predictive Control for Spacecraft Maneuvering,}
  \emph{IEEE Control Systems Letters}, Vol.~3, No.~3, 2019, pp. 649--654.

\bibitem[{Zinage and Bakolas(2021)}]{zinage2021far}
Zinage, V., and Bakolas, E., \enquote{Far-Field Minimum-Fuel Spacecraft
  Rendezvous using Koopman Operator and l 2/l 1 Optimization,} \emph{2021
  American Control Conference (ACC)}, IEEE, 2021, pp. 2992--2997.

\bibitem[{Massioni et~al.(2011)Massioni, Ankersen, and
  Verhaegen}]{massioni2011matching_pursuit}
Massioni, P., Ankersen, F., and Verhaegen, M., \enquote{A matching pursuit
  algorithm approach to chaser-target formation flying problems,} \emph{IEEE
  Transactions on Control Systems Technology}, Vol.~20, No.~2, 2011, pp.
  513--519.
\newblock \doi{10.1109/TCST.2011.2130526}.

\bibitem[{Lion and Handelsman(2012)}]{lion1968primer_calculus_of_variations}
Lion, P., and Handelsman, M., \enquote{Primer vector on fixed-time impulsive
  trajectories.} \emph{AIAA Journal}, Vol.~6, No.~1, 2012, pp. 127--132.
\newblock \doi{10.2514/3.4452}.

\bibitem[{Arzelier et~al.(2013)Arzelier, Louembet, Rondepierre, and
  Kara-Zaitri}]{arzelier2013new_conv4}
Arzelier, D., Louembet, C., Rondepierre, A., and Kara-Zaitri, M., \enquote{A
  new mixed iterative algorithm to solve the fuel-optimal linear impulsive
  rendezvous problem,} \emph{Journal of Optimization Theory and Applications},
  Vol. 159, No.~1, 2013, pp. 210--230.
\newblock \doi{10.1007/s10957-013-0282-z}.

\bibitem[{PRUSSING(1969)}]{prussing1969illustration}
PRUSSING, J.~E., \enquote{Illustration of the primer vector in time-fixed,
  orbit transfer.} \emph{AIAA Journal}, Vol.~7, No.~6, 1969, pp. 1167--1168.
\newblock \doi{10.2514/3.5297}.

\bibitem[{Prussing(1969)}]{prussing1969optimal_conv2}
Prussing, J.~E., \enquote{Optimal four-impulse fixed-time rendezvous in the
  vicinity of a circular orbit.} \emph{AIAA Journal}, Vol.~7, No.~5, 1969, pp.
  928--935.
\newblock \doi{10.2514/3.5246}.

\bibitem[{Carter and Brient(1995)}]{carter1995linearized_conv3}
Carter, T., and Brient, J., \enquote{Linearized impulsive rendezvous problem,}
  \emph{Journal of Optimization Theory and Applications}, Vol.~86, No.~3, 1995,
  pp. 553--584.
\newblock \doi{10.1007/BF02192159}.

\bibitem[{Lu and Liu(2013)}]{ping_lu2013autonomous_cone_2}
Lu, P., and Liu, X., \enquote{Autonomous trajectory planning for rendezvous and
  proximity operations by conic optimization,} \emph{Journal of Guidance,
  Control, and Dynamics}, Vol.~36, No.~2, 2013, pp. 375--389.
\newblock \doi{10.2514/1.58436}.

\bibitem[{Liu and Lu(2013)}]{liu2013robust}
Liu, X., and Lu, P., \enquote{Robust trajectory optimization for highly
  constrained rendezvous and proximity operations,} \emph{AIAA Guidance,
  Navigation, and Control (GNC) Conference}, 2013, p. 4720.
\newblock \doi{10.2514/6.2013-4720}.

\bibitem[{Liu and Lu(2014)}]{ping_liu2014solving}
Liu, X., and Lu, P., \enquote{Solving nonconvex optimal control problems by
  convex optimization,} \emph{Journal of Guidance, Control, and Dynamics},
  Vol.~37, No.~3, 2014, pp. 750--765.
\newblock \doi{10.2514/1.62110}.

\bibitem[{Dueri et~al.(2017)Dueri, Acıkmese, Scharf, and
  Harris}]{p:acikmese2017}
Dueri, D., Acıkmese, B., Scharf, D.~P., and Harris, M.~W., \enquote{Customized
  Real-Time Interior-Point Methods for Onboard Powered-Descent Guidance,}
  \emph{Journal of Guidance, Control, and Dynamics}, Vol.~40, No.~2, 2017, pp.
  197--212.
\newblock \doi{10.2514/1.G001480}.

\bibitem[{Karlgaard(2006)}]{huber_karlgaard2006robust_huber}
Karlgaard, C.~D., \enquote{Robust rendezvous navigation in elliptical orbit,}
  \emph{Journal of Guidance, Control, and Dynamics}, Vol.~29, No.~2, 2006, pp.
  495--499.
\newblock \doi{10.2514/1.19148}.

\bibitem[{Singla et~al.(2006)Singla, Subbarao, and
  Junkins}]{adpative_singla2006adaptive}
Singla, P., Subbarao, K., and Junkins, J.~L., \enquote{Adaptive output feedback
  control for spacecraft rendezvous and docking under measurement uncertainty,}
  \emph{Journal of Guidance, Control, and Dynamics}, Vol.~29, No.~4, 2006, pp.
  892--902.
\newblock \doi{10.2514/1.17498}.

\bibitem[{D’Amico et~al.(2013)D’Amico, Ardaens, Gaias, Benninghoff,
  Schlepp, and J{\o}rgensen}]{d_amico_2013noncooperative}
D’Amico, S., Ardaens, J.-S., Gaias, G., Benninghoff, H., Schlepp, B., and
  J{\o}rgensen, J., \enquote{Noncooperative rendezvous using angles-only
  optical navigation: system design and flight results,} \emph{Journal of
  Guidance, Control, and Dynamics}, Vol.~36, No.~6, 2013, pp. 1576--1595.
\newblock \doi{10.2514/1.59236}.

\bibitem[{Youmans and Lutze(1998)}]{neural_youmans1998neural}
Youmans, E.~A., and Lutze, F.~H., \enquote{Neural network control of space
  vehicle intercept and rendezvous maneuvers,} \emph{Journal of Guidance,
  Control, and Dynamics}, Vol.~21, No.~1, 1998, pp. 116--121.
\newblock \doi{10.2514/2.4206}.

\bibitem[{Gao et~al.(2009)Gao, Yang, and Shi}]{gao2009multi_h_infinity}
Gao, H., Yang, X., and Shi, P., \enquote{Multi-objective robust $H_\infty$
  Control of spacecraft rendezvous,} \emph{IEEE Transactions on Control Systems
  Technology}, Vol.~17, No.~4, 2009, pp. 794--802.
\newblock \doi{10.1109/TCST.2008.2012166}.

\bibitem[{Yao et~al.(2010)Yao, Xie, and He}]{yao2010flyaround_relative}
Yao, Y., Xie, R., and He, F., \enquote{Flyaround orbit design for autonomous
  rendezvous based on relative orbit elements,} \emph{Journal of Guidance,
  Control, and Dynamics}, Vol.~33, No.~5, 2010, pp. 1687--1692.
\newblock \doi{10.2514/1.48494}.

\bibitem[{Luo et~al.(2007)Luo, Tang, and Lei}]{luo2007optimal}
Luo, Y.-Z., Tang, G.-J., and Lei, Y.-J., \enquote{Optimal multi-objective
  linearized impulsive rendezvous,} \emph{Journal of Guidance, Control, and
  Dynamics}, Vol.~30, No.~2, 2007, pp. 383--389.
\newblock \doi{10.2514/1.21433}.

\bibitem[{Boyd and Vandenberghe(2004)}]{boyd2004convex}
Boyd, S., and Vandenberghe, L., \emph{Convex optimization}, Cambridge
  university press, 2004.

\bibitem[{Grant and Boyd(2014)}]{grant2014cvx_boyd}
Grant, M., and Boyd, S., \enquote{CVX: Matlab software for disciplined convex
  programming, version 2.1,} , 2014.

\bibitem[{Foucart and Rauhut(2013)}]{b:Foucart2013}
Foucart, S., and Rauhut, H., \emph{A Mathematical Introduction to Compressive
  Sensing}, Birkhauser Basel, 2013.
\newblock \doi{10.1007/978-0-8176-4948-7}.

\bibitem[{Daubechies et~al.(2010)Daubechies, DeVore, Fornasier, and
  G{\"u}nt{\"u}rk}]{robust_efficient_irls_2010iteratively}
Daubechies, I., DeVore, R., Fornasier, M., and G{\"u}nt{\"u}rk, C.~S.,
  \enquote{Iteratively reweighted least squares minimization for sparse
  recovery,} \emph{Communications on Pure and Applied Mathematics: A Journal
  Issued by the Courant Institute of Mathematical Sciences}, Vol.~63, No.~1,
  2010, pp. 1--38.
\newblock \doi{10.1002/cpa.20303}.

\bibitem[{Candes and Tao(2005)}]{candes2005decoding_irls}
Candes, E.~J., and Tao, T., \enquote{Decoding by linear programming,}
  \emph{IEEE Transactions on Information Theory}, Vol.~51, No.~12, 2005, pp.
  4203--4215.
\newblock \doi{10.1109/TIT.2005.858979}.

\bibitem[{Yamanaka and Ankersen(2002)}]{yamanaka2002new_dynamics_state_trans}
Yamanaka, K., and Ankersen, F., \enquote{New state transition matrix for
  relative motion on an arbitrary elliptical orbit,} \emph{Journal of Guidance,
  Control, and Dynamics}, Vol.~25, No.~1, 2002, pp. 60--66.
\newblock \doi{10.2514/2.4875}.

\bibitem[{Wang et~al.(2013)Wang, Wang, and Xu}]{block_wang2013recovery}
Wang, Y., Wang, J., and Xu, Z., \enquote{On recovery of block-sparse signals
  via mixed $\ell_2/\ell_q (0< q\leq 1)$ norm minimization,} \emph{EURASIP
  Journal on Advances in Signal Processing}, Vol. 2013, No.~1, 2013, p.~76.

\bibitem[{Zhou et~al.(2011)Zhou, Lin, and Duan}]{zhou2011lyapunov_gto}
Zhou, B., Lin, Z., and Duan, G.-R., \enquote{Lyapunov differential equation
  approach to elliptical orbital rendezvous with constrained controls,}
  \emph{Journal of Guidance, Control, and Dynamics}, Vol.~34, No.~2, 2011, pp.
  345--358.
\newblock \doi{10.2514/1.52372}.

\bibitem[{Arzelier et~al.(2016)Arzelier, Br{\'e}hard, Deak, Joldes, Louembet,
  Rondepierre, and Serra}]{arzelier2016linearized_compare}
Arzelier, D., Br{\'e}hard, F., Deak, N., Joldes, M., Louembet, C., Rondepierre,
  A., and Serra, R., \enquote{Linearized impulsive fixed-time fuel-optimal
  space rendezvous: A new numerical approach,} \emph{International Federation
  of Automatic Control}, Vol.~49, No.~17, 2016, pp. 373--378.
\newblock \doi{10.1016/j.ifacol.2016.09.064}.

\bibitem[{Labourdette et~al.(2008)Labourdette, Julien, Chemama, and
  Carbonne}]{labourdette2008atv}
Labourdette, P., Julien, E., Chemama, F., and Carbonne, D., \enquote{ATV Jules
  Verne mission maneuver plan,} \emph{International Symposium on Space Flight
  Dynamics, Toulouse, France}, 2008.

\end{thebibliography}

\end{document}